\newif\ifprocs
\procsfalse

\ifprocs
\documentclass[twoside,leqno,twocolumn]{article}  
\usepackage{ltexpprt} 
\usepackage{balance}
\else
\documentclass[11pt,fleqn]{article}
\usepackage{fullpage}
\usepackage{amsthm}
\fi

\usepackage{amsmath,amssymb}
\usepackage[boxed]{algorithm}
\usepackage[noend]{algorithmic}
\usepackage{bbm}
\usepackage{graphicx}
\usepackage{subfigure}
\usepackage{xspace}
\usepackage{color}

\usepackage{ifpdf}
\ifpdf    
\usepackage{hyperref}
\else    
\usepackage[hypertex]{hyperref}
\fi

\ifprocs
\newtheorem{definition}[Definition]{Definition}
\newtheorem{assumption}[Definition]{Assumption}
\newtheorem{notation}[Definition]{Notation}
\newtheorem{claim}[lemma]{Claim}
\newtheorem{remark}[Remark]{Remark}
\else
\newtheorem{theorem}{Theorem}[section]
\newtheorem{lemma}[theorem]{Lemma}

\newtheorem{corollary}[theorem]{Corollary}
\newtheorem{definition}[theorem]{Definition}
\newtheorem{proposition}[theorem]{Proposition}

\newtheorem{claim}[theorem]{Claim}
\newtheorem{assumption}[theorem]{Assumption}

\newtheorem{notation}[theorem]{Notation}
\fi

\ifprocs
\newcommand{\proofof}[1]{of #1}
\else
\newcommand{\proofof}[1]{Proof of #1}
\renewcommand{\qedhere}{}
\fi

\newcommand {\ignore} [1] {}
\def \iter {10 \log_b \calD}

\DeclareMathOperator{\supp}{supp}
\DeclareMathOperator{\diam}{diam}
\DeclareMathOperator*{\EX}{{\mathbb E}}
\DeclareMathOperator{\polylog}{polylog}

\newcommand{\etal}{{\em et al.\ }\xspace}

\providecommand{\card}[1]{\lvert#1\rvert}

\providecommand{\aset}[1]{\{#1\}}
\providecommand{\eqdef}{:=}

\def\compactify{\itemsep=0pt \topsep=0pt \partopsep=0pt \parsep=0pt}

\newcommand{\calD}{{\mathcal D}}
\newcommand{\far}{{\text{far}}}
\newcommand{\near}{{\text{near}}}
\newcommand{\fin}{{\text{fin}}}
\newcommand{\longSub}{{\text{long}}}
\DeclareMathOperator{\texp}{Texp}

\title{Cutting corners cheaply, or how to remove Steiner points%
  \thanks{A preliminary version appeared in Proceedings of the 25th Annual ACM-SIAM
Symposium on Discrete Algorithms, 2014}}
\author{Lior Kamma%
\thanks{This work was supported in part by the Israel Science Foundation (grant 
\#897/13), the US-Israel BSF (grant \#2010418), and by the Citi Foundation.
Part of this work was done while visiting Microsoft Research New England.
Email: \texttt{\{lior.kamma,robert.krauthgamer\}@weizmann.ac.il}
}
\\ The Weizmann Institute
\and Robert Krauthgamer\footnotemark[\value{footnote}]
\\ The Weizmann Institute
\and Huy L. Nguy$\tilde{\mbox{\^{e}}}$n%
\thanks{This work was supported in part by NSF CCF 0832797, and a Gordon Wu Fellowship.
Part of this work was done while interning at Microsoft Research New England.
Email: \texttt{hlnguyen@princeton.edu}
}
\\ Princeton University
}

\ifprocs
\date{}
\fi

\begin{document}

\maketitle

\begin{abstract}
\ifprocs
  \small
\fi
Our main result is that the Steiner Point Removal (SPR) problem
can always be solved with polylogarithmic distortion,
which answers in the affirmative 
a question posed by Chan, Xia, Konjevod, and Richa (2006).
Specifically, we prove that for every edge-weighted graph $G = (V,E,w)$
and a subset of terminals $T \subseteq V$, 
there is a graph $G'=(T,E',w')$ that is isomorphic to a minor of $G$,
such that for every two terminals $u,v\in T$,
the shortest-path distances between them in $G$ and in $G'$ satisfy
$d_{G,w}(u,v) \le d_{G',w'}(u,v) \le O(\log^5|T|) \cdot d_{G,w}(u,v)$.
Our existence proof actually gives a randomized polynomial-time algorithm.

\ifprocs
  \small
\fi
Our proof features a new variant of metric decomposition.
It is well-known that every finite metric space $(X,d)$ admits 
a $\beta$-separating decomposition for $\beta=O(\log \card{X})$, 
which means that for every $\Delta>0$ there is a randomized 
partitioning of $X$ into clusters of diameter at most $\Delta$,
satisfying the following separation property:
for every $x,y \in X$, the probability they lie in different clusters 
of the partition is at most $\beta\,d(x,y)/\Delta$. 
We introduce an additional requirement in the form of a tail bound:
for every shortest-path $P$ of length $d(P) \leq \Delta/\beta$, 
the number of clusters of the partition that meet the path $P$, denoted $Z_P$,
satisfies $\Pr[Z_P > t] \le 2e^{-\Omega(t)}$ for all $t>0$.
\end{abstract}

\section{Introduction}

\emph{Graph compression} describes the transformation of a given graph $G$ 
into a small graph $G'$ that preserves certain features (quantities) of $G$,
such as distances or cut values.
Notable examples for this genre include graph spanners, distance oracles,
cut sparsifiers, and spectral sparsifiers, 
see e.g.\ \cite{PS89,TZ05,BK96,BSS08} and references therein.
The algorithmic utility of such graph transformations is clear  -- 
once the ``compressed'' graph $G'$ is computed as a preprocessing step,
further processing can be performed on $G'$ instead of on $G$,
using less resources like runtime and memory, 
or achieving better accuracy (when the solution is approximate). 
See more in Section \ref{sec:related}.

Within this context, we study \emph{vertex-sparsification},
where $G$ has a designated subset of vertices $T$,
and the goal is to reduce the number of vertices in the graph
while maintaining certain properties of $T$.
A prime example for this genre is vertex-sparsifiers
that preserve terminal versions of (multicommodity) cut and flow problems, 
a successful direction that was initiated by Moitra \cite{Moitra09} 
and extended in several followups \cite{LM10,CLLM10,MM10,EGKRTT10,Chuzhoy12}.
Our focus here is different, on preserving distances, 
a direction that was seeded by Gupta \cite{Gupta01} more than a decade ago.

Throughout the paper, all graphs are undirected and all edge weights are positive.

\paragraph{Steiner Point Removal (SPR).}
Let $G=(V,E,w)$ be an edge-weighted graph%
and let $T=\aset{t_1,\ldots,t_k}\subseteq V$ be a designated set of $k$ terminals.
Here and throughout, $d_{G,w}(\cdot,\cdot)$ denotes the shortest-path metric 
between vertices of $G$ according to the weights $w$.
The \emph{Steiner Point Removal} problem asks to construct on the terminals
a new graph $G'=(T,E',w')$ such that
(i)  
distances between the terminals are \emph{distorted} at most 
by factor $\alpha\ge 1$, formally
$$
\forall u,v\in T,\qquad d_{G,w}(u,v) \le d_{G',w'}(u,v) \le \alpha\cdot d_{G,w}(u,v);
$$
and (ii) the graph $G'$ is (isomorphic to) a minor of $G$.
This formulation of the SPR problem was proposed by
Chan, Xia, Konjevod, and Richa \cite[Section 5]{CXKR06}
who posed the problem of bounding the distortion $\alpha$
(existentially and/or using an efficient algorithm).
Our main result is to answer their open question.

Requirement (ii) above expresses structural similarity between $G$ and $G'$; 
for instance, if $G$ is planar then so is $G'$.
The SPR formulation above actually came about as a generalization 
to a result of Gupta \cite{Gupta01}, which
asserts that if $G$ is a tree, there exists a tree $G'$, which preserves terminal distances with distortion $\alpha=8$.
Later Chan \etal \cite{CXKR06} observed that this same $G'$ 
is actually a minor of the original tree $G$,
and proved the factor of $8$ to be tight.
The upper bound for trees was later extended by Basu and Gupta \cite{BG08},
who achieve distortion $\alpha=O(1)$ for the larger class of outerplanar graphs.

\paragraph{How to construct minors.}

We now describe a general methodology that is natural for the SPR problem.
The first step constructs a minor $G'$ with vertex set $T$, 
but without any edge weights,
and is prescribed by Definition \ref{defn:tcm}.
The second step determines edge weights $w'$ 
such that $d_{G',w'}$ dominates $d_{G,w}$ on the terminals $T$,
and is given in Definition \ref{defn:sr}.
These steps are illustrated in Figure \ref{fig:defs}.
Our definitions are actually more general (anticipating the technical sections),
and consider $G'$ whose vertex set is sandwiched between $T$ and $V$.
\begin{definition}
A \emph{partial partition} of a set $V$ is a collection $V_1,\ldots,V_k$ 
of pairwise disjoint subsets of $V$, referred to as {\em clusters}.
\end{definition}

\begin{figure}[ht]
  \begin{center}
    \label{fig:defs}
    \includegraphics[scale=0.45]{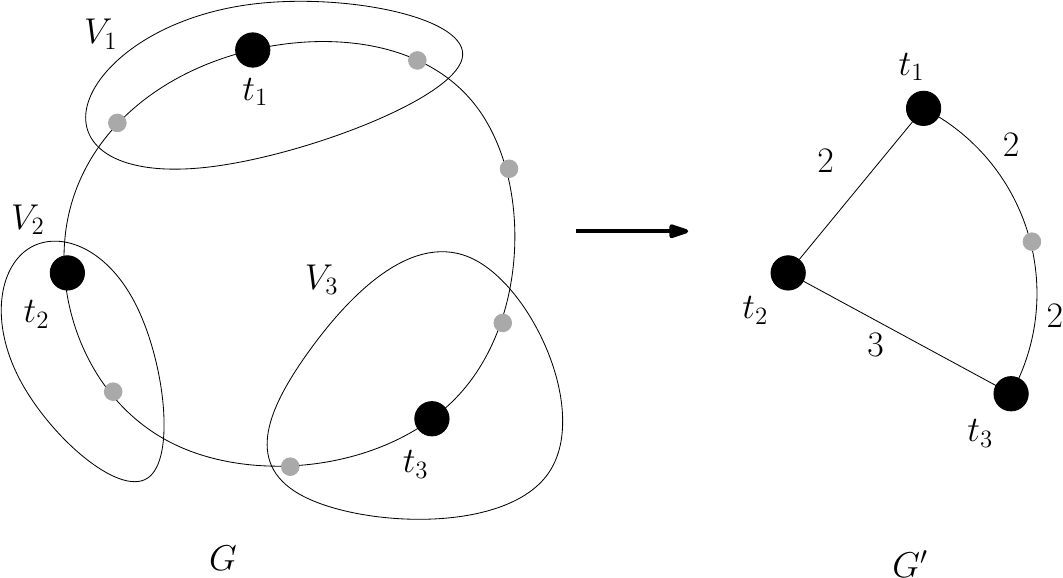}
    \caption{The graph $G$, a $9$-cycle with unit edge weights, 
is depicted on the left with $3$ terminals and disjoint subsets $V_1,V_2,V_3$. 
Its terminal-centered minor $G'$ and the standard-restriction 
edge weights are shown on the right.}
  \end{center}
\hrule
\end{figure}

\begin{definition}[Terminal-Centered Minor] \label{defn:tcm}
Let $G=(V,E)$ be a graph with $k$ terminals $T=\aset{t_1,\ldots,t_k}$,
and let $V_1,\ldots,V_k$ be a partial partition of $V$, 
such that each induced subgraph $G[V_j]$ is connected and contains $t_j$.
The graph $G'=(V',E')$ obtained by contracting each $G[V_j]$ 
into a single vertex that is identified with $t_j$, 
is called the {\em terminal-centered minor} of $G$ induced by $V_1,\ldots,V_k$.
\end{definition}

By identifying the ``contracted super-node'' $V_j$ with $t_j$, 
we may think of the vertex-set $V'$ as containing $T$
and (possibly) some vertices from $V\setminus T$,
which implies $V'\subset V$.
A terminal-centered minor $G'$ of $G$ 
can also be described by a mapping $f:V\to T\cup\{\bot\}$, 
such that $f|_T\equiv\mathrm{id}$ 
and $f^{-1}(\{t_j\})$ is connected in $G$ for all $j \in [k]$.
Indeed, simply let $V_j = f^{-1}(\{t_j\})$ for all $j \in [k]$, 
and thus $V \setminus \left(\cup_jV_j\right) = f^{-1}(\{\bot\})$. 

\begin{definition}[Standard Restriction] \label{defn:sr}
Let $G=(V,E,w)$ be an edge-weighted graph with terminal set $T$,
and let $G'=(V',E')$ be a terminal-centered minor of $G$.
(Recall we view $V'\subset V$.)
The {\em standard restriction} of $w$ to $G'$ is the edge weight $w'$
given by the respective distances in $G$, formally
$$\forall(x,y)\in E',\qquad w'_{xy}\eqdef d_{G,w}(x,y).$$
\end{definition}

This edge weight $w'$ is optimal in the sense that 
$d_{G',w'}$ dominates $d_{G,w}$ (where it is defined, i.e., on $V'$),
and the weight of each edge $(x,y)\in E'$ is minimal
under this domination condition.

\subsection{Main Result} \label{sec:results}

Our main result below gives an efficient algorithm 
that achieves $\polylog(k)$ distortion for the SPR problem.
Its proof spans Sections \ref{sec:TCM} and \ref{sec:proofDiscard},
though the former contains the heart of the matter.

\begin{theorem} \label{thm:main}
Let $G=(V,E,w)$ be an edge-weighted graph with $k$ terminals $T\subseteq V$.
Then there exists a terminal-centered minor $G'=(T,E',w')$ of $G$ that attains distance distortion $O(\log^5 k)$, i.e., 
$$
  \forall u,v \in T,\qquad
  1 \le \frac{d_{G',w'}(u,v)}{d_{G,w}(u,v)}  \le O(\log^5 k).
$$
Moreover, $w'$ is the standard restriction of $w$,
and $G'$ is computable in randomized polynomial time.
\end{theorem}

This theorem answers a question of Chan \etal \cite{CXKR06}. %
The only distortion lower bound known for general graphs 
is a factor of $8$ (which actually holds for trees) \cite{CXKR06},
and thus it remains a challenging open question whether 
$O(1)$ distortion can be achieved in general graphs.

Our proof of Theorem \ref{thm:main} begins similarly to the proof of Englert \etal \cite{EGKRTT10}, 
by iterating over the ``distance scales'' $2^i$, 
going from the smallest distance $d_{G,w}(u,v)$ among all terminals $u,v \in T$,
towards the largest such distance.
Each iteration $i$ first employs a ``stochastic decomposition'',
which is basically a randomized procedure that finds 
clusters of $V$ whose diameter is at most $2^i$.
Then, some clusters are contracted to a nearby terminal,
which must be ``adjacent'' to the cluster;
this way, the current graph is a minor of the previous iteration's graph, 
and thus also of the initial $G$.
After iteration $i$ is executed, 
we roughly expect ``neighborhoods'' of radius proportional to $2^i$ 
around the terminals to be contracted.
As $i$ increases, these neighborhoods get larger
until eventually all the vertices are contracted into terminals,
at which point the edge weights are set according to the standard restriction.
To eventually get a minor, 
it is imperative that every contracted region is connected.
To guarantee this, we perform the iteration $i$ decomposition in the graph 
resulting from previous iterations' contractions (rather than the initial $G$),
which introduces further dependencies between the iterations.

The main challenge is to control the distortion,
and this is where we crucially deviate from \cite{EGKRTT10}
(and differ from all previous work).
In their randomized construction of a minor $G'$, 
for every two terminals $u,v\in T$ it is shown that $G'$ contains
a $uv$-path of \emph{expected length} at most $O(\log k) d_G(u,v)$.
Consequently, they design a \emph{distribution} $D$ over minors $G'$,
such that the stretch $d_{G'}(u,v)/d_G(u,v)$ between any $u,v\in T$
has expectation at most $O(\log k)$. %
Note, however, that it is possible that no $G'\in\supp(D)$ achieves a low stretch 
simultaneously for all $u,v\in T$.
In contrast, in our randomized construction of $G'$,
the stretch between $u,v\in T$ is polylogarithmic
\emph{with high probability}, say at least $1-1/k^3$.
Applying a simple union bound over the $\binom{k}{2}$ terminal pairs,
we can then obtain a single graph $G'$ achieving a polylogarithmic distortion.
Technically, these bounds follow by fixing in $G$ 
a shortest-path $P$ between two terminals $u,v\in T$,
and then tracking the execution of the randomized algorithm
to analyze how the path $P$ evolves into a $uv$-path $P'$ in $G'$.
In \cite{EGKRTT10}, the length of $P'$ is analyzed in expectation,
which by linearity of expectation, 
follows from analyzing the case where $P$ consists of a single edge;
In contrast, we provide for $P$ a high-probability bound,
which inevitably must consider (anti)correlations along the path.

The next section features a new tool that we developed in our quest 
for high-probability bounds, and which may be of independent interest.
For the sake of clarity, we provide below  
a vanilla version that excludes technical complications 
such as terminals, strong diameter, and consistency between scales.
The proof of Theorem \ref{thm:main} actually does require these complications, 
and thus cannot use the generic form described below.

\subsection{A Key Technique: Metric Decomposition with Concentration} \label{sec:techniques}

\paragraph{Metric decomposition.}

Let $(X,d)$ be a metric space, and let $\Pi$ be a partition of $X$.
Every $S \in \Pi$ is called a {\em cluster},
and for every $x\in X$, we use $\Pi(x)$ to denote the unique cluster 
$S \in \Pi$ such that $x \in S$.
In general, a stochastic decomposition of the metric $(X,d)$ 
is a distribution $\mu$ over partitions of $X$,
although we usually impose additional requirements.
The following definition is perhaps the most basic version,
often called a separating decomposition or a Lipschitz decomposition.

\begin{definition}\label{def:beta}
A metric space $(X,d)$ is called {\em $\beta$-decomposable} 
if for every $\Delta>0$, there is a probability distribution $\mu$ over 
partitions of $X$, satisfying the following requirements:
\begin{enumerate} \compactify
\renewcommand{\theenumi}{\emph{(\alph{enumi})}}
\item \label{it:DiameterBound}
Diameter bound: for all $\Pi \in \supp(\mu)$ and all $S\in \Pi$, 
  \ $\diam(S) \le \Delta$. 
\item \label{it:SeparatingProb}
Separation probability: for all $x,y \in X$,
  \ $\displaystyle\Pr_{\Pi \sim \mu}[\Pi(x) \ne \Pi(y)] \le \tfrac{\beta d(x,y)}{\Delta}$.
\end{enumerate}
\end{definition}
We note that the diameter bound holds with respect to distances in $X$;
in the case of a shortest-path metric in a graph, 
this is known as a {\em weak-diameter} bound.

Bartal \cite{Bartal96} proved that every $n$-point metric is 
$O(\log n)$-decomposable, and that this bound is tight. 
We remark that by now there is a rich literature on metric decompositions,
and different variants of this notion may involve terminals,
or (in a graphical context) connectivity requirements inside each cluster,
see e.g.\ \cite{LS93,Bartal96, CKR01, FRT04, Bartal04, LN05, GNR10, EGKRTT10, MN07, AGMW10, KR11}.

\paragraph{Degree of separation.}
Let $P=(x_0,x_1,\ldots,x_\ell)$ be a {\em shortest path},
i.e., a sequence of points in $X$ such that
$\sum_{i\in[\ell]}{d(x_{i-1},x_{i})} = d(x_0,x_\ell).$
We denote its length by $d(P)\eqdef d(x_0,x_\ell)$,
and say that $P$ {\em meets} a cluster $S\subseteq X$ 
if $S \cap P\neq \emptyset$. 
Given a partition $\Pi$ of $X$, 
define the \emph{degree of separation} $Z_P(\Pi)$ as
the number of different clusters in the partition $\Pi$ that meet $P$. 
Formally,
\begin{equation}
Z_P(\Pi) 
  \eqdef \sum_{S \in \Pi} \mathbbm{1}_{\aset{\text{$P$ meets $S$}} }.
\label{eq:dosdef}
\end{equation}

Throughout, we omit the partition $\Pi$ when it is clear from the context.
When we consider a random partition $\Pi \sim\mu$, 
the corresponding $Z_P=Z_P(\Pi)$ is actually a random variable.
If this distribution $\mu$ satisfies 
requirement \ref{it:SeparatingProb} of Definition \ref{def:beta}, then 
\begin{align}
\ifprocs
  \label{eq:ExpectedCount}
  \EX_{\Pi\sim\mu}[Z_P] 
  & \leq 1 + \sum_{i \in [\ell]} \Pr_{\Pi\sim\mu}[\Pi(x_{i-1})\neq \Pi(x_i)] 
  \\ \nonumber
  & \leq 1 + \sum_{i \in [\ell]} \frac{\beta d(x_{i-1},x_i)}{\Delta} 
  =   1 + \frac{\beta d(P)}{\Delta}.
\else
  \label{eq:ExpectedCount}
  \EX_{\Pi\sim\mu}[Z_P] 
  \leq 1 + \sum_{i \in [\ell]} \Pr_{\Pi\sim\mu}[\Pi(x_{i-1})\neq \Pi(x_i)] 
  \leq 1 + \sum_{i \in [\ell]} \frac{\beta d(x_{i-1},x_i)}{\Delta}
  =   1 + \frac{\beta d(P)}{\Delta}.
\fi
\end{align}
But what about the concentration of $Z_P$?
More precisely, can every finite metric be decomposed, such that 
every shortest path $P$ admits a tail bound on its degree of separation $Z_P$?

\paragraph{A tail bound.}
We answer this last question in the affirmative using the following theorem.
We prove it, or actually a stronger version that does involve terminals,
in Section \ref{sec:main2}.

\begin{theorem}\label{thm:main2}
For every $n$-point metric space $(X,d)$ and every $\Delta > 0$ 
there is a probability distribution $\mu$ over partitions of $X$
that satisfies, for $\beta = O(\log n)$,
requirements \ref{it:DiameterBound}-\ref{it:SeparatingProb} 
of Definition \ref{def:beta}, and furthermore
\begin{enumerate} \compactify
\renewcommand{\theenumi}{\emph{(\alph{enumi})}}
\addtocounter{enumi}{2}
\item \label{it:DegreeSeparation}
Degree of separation: For every shortest path $P$ of length $d(P)\leq \frac{\Delta}{\beta}$, 
\begin{equation}
  \label{eq:SimpleTail}
  \forall t \ge 1, \qquad 
  \Pr_{\Pi\sim\mu}\left[Z_P >  t \right] \le 2e^{- \Omega(t)}.
\end{equation}
\end{enumerate}
\end{theorem}
The tail bound \eqref{eq:SimpleTail} can be compared to a naive estimate 
that holds for every $\beta$-decomposition $\mu$:
using \eqref{eq:ExpectedCount} we have $\EX[Z_P]\le 2$,
and then by Markov's inequality $\Pr[Z_P \ge t] \le 2/t$.

We remark that for general metric spaces, it is known that $\beta = O(\log n)$ is tight \cite{Bartal96}. However, for requirements \ref{it:DiameterBound}-\ref{it:SeparatingProb} of Definition \ref{def:beta} several decompositions are known to have better values of $\beta$ 
for special families of metric spaces 
(e.g., metrics induced by planar graphs \cite{KPR93}). 
We leave it open whether for these families
the bounds of Theorem~\ref{thm:main2} can be improved, say to $\beta=O(1)$.

\subsection{Related Work} \label{sec:related}

\paragraph{Applications.}
Vertex-sparsification, and the ``graph compression'' approach in general, 
is obviously beneficial when $G'$ can be computed from $G$ very efficiently, 
say in linear time, 
and then $G'$ may be computed on the fly rather than in advance.
But compression may be valuable also in scenarios
that require the storage of many graphs, like archiving and backups,
or rely on low-throughput communication, like distributed or remote processing.
For instance, the succinct nature of $G'$ may be indispensable for computations
performed frequently, say on a smartphone,
with preprocessing done in advance on a powerful machine.

We do not have new theoretical applications that leverage our SPR result,
although we anticipate these will be found later. 
Either way, we believe this line of work will prove technically productive,
and may influence, e.g., work on metric embeddings 
and on approximate min-cut/max-flow theorems.

\paragraph{Probabilistic SPR.}
Here, the objective is not to find a single graph $G'=(T,E',w')$,
but rather a distribution $D$ over graphs $G'=(T,E',w')$,
such that every graph $G' \in \supp(D)$ is isomorphic to a minor of $G$
and its distances $d_{G',w'}$ dominate $d_{G,w}$ (on $T \times T$),
and such that the distortion inequalities hold in expectation, that is,
$$\forall u,v \in T,\qquad \EX_{G'\sim D}[d_{G',w'}(u,v)] \le \alpha\cdot d_{G,w}(u,v).$$
This problem, first posed by Chan \etal in \cite{CXKR06}, was answered by Englert \etal in \cite{EGKRTT10}
with $\alpha = O(\log |T|)$.

\paragraph{Distance Preserving Minors.} 
This problem is a relaxation of SPR in which the minor $G'$ may contain a few non-terminals,
while preserving terminal distances exactly.
Formally, the objective is to find a small graph $G' = (V',E',w')$ such that
(i) $G'$ is isomorphic to a minor of $G$;
(ii) $T \subseteq V' \subseteq V$; and
(iii) for every $u,v \in T$,\ $d_{G',w'}(u,v) = d_{G,w}(u,v)$.
This problem was originally defined by Krauthgamer, Nguy$\tilde{\mbox{\^{e}}}$n and Zondiner \cite{KNZ14},
who showed an upper bound $|V'|\leq O(|T|^4)$ for general graphs,
and a lower bound of $\Omega(|T|^2)$ that holds even for planar graphs.

\section{Metric Decomposition with Concentration} \label{sec:main2}

In this section we prove a slightly stronger result than that of Theorem~\ref{thm:main2},
stated as Theorem \ref{thm:main3} below. 
Let $(X,d)$ be a metric space, and let $\aset{t_1,\ldots, t_k}\subseteq X$ 
be a designated set of terminals.
Recall that a partial partition $\Pi$ of $X$ 
is a collection of pairwise disjoint subsets of $X$. 
For a shortest path $P$ in $X$,
define $Z_P=Z_P(\Pi)$ using Eqn.~\eqref{eq:dosdef},
which is similar to before, except that now $\Pi$ is a partial partition.
We first extend Definition~\ref{def:beta}.

\begin{definition}
We say that $X$ is {\em $\beta$-terminal-decomposable with concentration} if for every $\Delta>0$ there is a probability distribution $\mu$ over partial partitions of $X$, satisfying the following properties.
\begin{itemize} \compactify
	\item {\em Diameter Bound:} For all $\Pi \in \supp(\mu)$ and all $S \in \Pi$,\ $diam(S) \le \Delta$.
	\item {\em Separation Probability:} For every $x,y \in X$, $$\Pr_{\Pi \sim \mu}[\exists S \in \Pi \ \text{such that} \ |S \cap \{x,y\}| =1 \ ] \le \tfrac{\beta d(x,y)}{\Delta} \; .$$
	\item {\em Terminal Cover:} For all $\Pi \in \supp(\mu)$, we have \mbox{$T \subseteq \bigcup_{S \in \Pi}S$}.
	\item {\em Degree of Separation:} For every shortest path $P$ and every $t \ge 1$, 
          \begin{multline*}
            \Pr_{\Pi \sim \mu} \left[Z_P >  t \max\{\tfrac{\beta d(P)}{\Delta}, 1 \}\right] 
\ifprocs
            \\ 
\fi
            \leq O\left( \min \left\{ k\beta , \left\lceil \tfrac{\beta d(P)}{\Delta}\right\rceil \right\} \right)  \ e^{- \Omega(t)}.
          \end{multline*}

\end{itemize}
\end{definition}

\begin{theorem} \label{thm:main3}
Every finite metric space with $k$ terminals 
is $(4\log k)$-terminal-decomposable with concentration.
\end{theorem}

Define the {\em truncated exponential} with parameters $\lambda,\Delta>0$, 
denoted $\texp(\lambda,\Delta)$,
to be distribution given by the probability density function 
$g_{\lambda,\Delta}(x) = \frac{1}{\lambda(1 - e^{- \Delta / \lambda})}e^{- x/ \lambda}$
for all $x \in [0,\Delta)$. 

We are now ready to prove Theorem \ref{thm:main3}. For simplicity of notation, we prove the result with cluster diameter at most $2 \Delta$ instead of $\Delta$.
Fix a desired $\Delta>0$,
and set for the rest of the proof $\lambda\eqdef\frac{\Delta}{\log k}$ 
and $g \eqdef g_{\lambda,\Delta}$. 
For $x \in X$ and $r> 0$, we use the standard notation of a closed ball
$B(x,r) \eqdef \{y \in X:\ d(x,y) \le r\}$. 
We define the distribution $\mu$ via the following procedure that samples 
a partial partition $\Pi$ of $X$.
\begin{algorithm}[H]
\begin{algorithmic}[1]
\FOR{$j = 1,2,\ldots,k$}
\STATE choose $R_j \sim \texp(\lambda,\Delta)$ independently at random, and let $B_j = B(t_j,R_j)$.
\STATE set $S_j = B_j \setminus \bigcup_{m=1}^{j-1}B_m$.
\ENDFOR
\RETURN $\Pi = \{ S_1,\ldots,S_k\} \setminus \{\emptyset\}$.
\end{algorithmic}
\end{algorithm}
The diameter bound and terminal partition properties hold by construction. The proof of the separation event property is identical to the one in \cite[Section 3]{Bartal96}.
The following two lemmas prove the degree of separation property, 
which will conclude the proof of Theorem~\ref{thm:main3}. 
Fix a shortest path $P$ in $X$,
and let us assume that $t/2$ is a positive integer;
a general $t\ge1$ can be reduced to this case up to a loss
in the unspecified constant.

\begin{lemma} \label{l:shortShortest}
If $d(P) < \lambda$, then $\Pr[Z_P > t]  \le 2e^{ - \Omega(t)}$.
\end{lemma}
\begin{proof}
Split the $k$ terminals into 
$J_\far \eqdef \{ j \in [k]:\ d(t_j,P) > \Delta - 2\lambda \}$ 
and $J_\near \eqdef [k] \setminus J_\far$.
Define random variables 
$Z_\far\eqdef \#\aset{j \in J_\far:\ B_j \cap P \ne \emptyset}$ 
and $Z_\near\eqdef \#\aset{j \in J_\near:\ S_j \cap P \ne \emptyset}$.
Then $Z_P \le Z_\far+Z_\near$ and 
\begin{align*}
\ifprocs
\Pr[Z_P > t] &\le \Pr[Z_\far+Z_\near > t] \\
&\le \Pr[Z_\far>t/2] + \Pr[Z_\near> t/2].
\else
\Pr[Z_P > t] \le \Pr[Z_\far+Z_\near > t] \le \Pr[Z_\far>t/2] + \Pr[Z_\near> t/2].
\fi
\end{align*}
For every $j \in J_\far$, 
\begin{align*}
\ifprocs
\Pr[B_j \cap P \ne \emptyset] 
  &\le \Pr[R_j \geq \Delta - 2\lambda] \\
  &= \int_{\Delta- 2\lambda}^{\Delta}{g(x)dx} \\
	&= \frac{k}{k-1}(e^{- \frac{\Delta- 2\lambda}{\lambda}} - e^{- \frac{\Delta}{\lambda}}) \le \frac{8}{k},
\else
\Pr[B_j \cap P \ne \emptyset] 
  \le \Pr[R_j \geq \Delta - 2\lambda] 
  = \int_{\Delta- 2\lambda}^{\Delta}{g(x)dx} = \frac{k}{k-1}(e^{- \frac{\Delta- 2\lambda}{\lambda}} - e^{- \frac{\Delta}{\lambda}}) \le \frac{8}{k},
\fi
\end{align*}
and therefore $\EX[Z_\far] \le 8$.
Since $Z_\far$ is the sum of independent indicators, 
by the Chernoff bound, $\Pr[Z_\far> t/2] \le 2^{- t/2}$ 
for all $t\ge 32e$.
For smaller $t$, observe that $\Pr[Z_\far = 0] \ge (1-8/k)^k\ge \Omega(1)$, 
and thus for every $t \ge 1$ we have $\Pr[Z_\far> t/2] \le e^{-\Omega(t)}$.

Next, consider the balls among $\{B_j:\ j \in J_\near\}$ 
that have non-empty intersection with $P$. Let $m$ denote the number of such balls, 
and let $j_1< \ldots < j_m$ denote their indices.
In other words, we condition henceforth on an event $\mathcal E \in \{0,1\}^{J_{near}}$ 
that determines whether $R_j\ge d(t_j,P)$ occurs or not for each $j\in J_\near$. The indices of coordinates of $\mathcal E$ that are equal to $1$ are exactly $j_1, \ldots j_m$.
For $a \in [m]$, let $Y_a$ be the indicator variable for the event that the
ball $B_{j_a}$ does \emph{not} contain $P$. Note that since $\{R_j\}_{j \in [k]}$ are independent, then so are $\{Y_a\}_{a \in [m]}$. Then
\begin{align*}
\ifprocs
  \Pr\Big[Y_a = 1 \mid &{\mathcal E}\Big] 
  = \Pr\Big[P \not \subseteq B_{j_a} \mid P \cap B_{j_a} \ne \emptyset\Big] \\
  &\le \Pr\Big[R_j < d(t_{j_a}, P) + \lambda \mid R_j \ge d(t_{j_a},P)\Big] \\
  &\le \frac{1 - e^{-1}}{1 - e ^{- 2}} 
  \le \frac{3}{4}.
\else
  \Pr\Big[Y_a = 1 \mid {\mathcal E}\Big] 
  &= \Pr\Big[P \not \subseteq B_{j_a} \mid P \cap B_{j_a} \ne \emptyset\Big] \\
  &\le \Pr\Big[R_j < d(t_{j_a}, P) + \lambda \mid R_j \ge d(t_{j_a},P)\Big] 
  \le \frac{1 - e^{-1}}{1 - e ^{- 2}} 
  \le \frac{3}{4}.
\fi
\end{align*}
Having conditioned on $\mathcal E$,
the event $\aset{Z_\near>t/2}$ implies that $m > t/2$ and moreover, $Y_a=1$ for all $a \in [t/2]$,
and since $\{Y_a\}_{a \in [m]}$ are independent, 
$\Pr[Z_\near>t/2 \mid {\mathcal E}] \le (3/4)^{t/2} \le e^{-Ct}$
for an appropriate constant $C>0$. 
The last inequality holds for all such events $\mathcal E$ 
(with the same constant $C>0$),
and thus also without any such conditioning.

Altogether, we conclude that
$ 
  \Pr[Z_P > t] 
  \le 2e^{ - \Omega(t)}.
$
\end{proof}

\begin{lemma} \label{l:longShortest}
If $d(P) \ge \lambda$, then $\Pr[Z_P > t d(P) / \lambda]  \le O\left(\min\left\{ k \log k,\left\lceil \frac{d(P)}{\lambda}\right\rceil \right\}\right)  e^{ - \Omega(t)}$.
\end{lemma}
\begin{proof}
Treating $P$ as a continuous path, 
subdivide it into $r \eqdef \left\lceil d(P) / \lambda \right\rceil$ segments,
say segments of equal length that are (except for the last one)\
half open and half closed.
The induced subpaths $P_1,\ldots,P_r$ of $P$ 
are disjoint (as subsets of $X$) and have length at most $\lambda$ each,
though some of subpaths may contain only one or even zero points of $X$.
Writing $Z_P=\sum_{i\in[r]} Z_{P_i}$, we can
apply a union bound and then Lemma~\ref{l:shortShortest} on each $P_i$, 
to obtain
\begin{align*}
\ifprocs
\Pr[Z_P > t d(P) / \lambda] 
  &\le \Pr\Big[\exists i \in [r] \ \text{such that} \ Z_{P_i} > t/2\Big] \\
  &\le O \left(\left\lceil \frac{d(P)}{\lambda} \right\rceil \right) \cdot e^{- \Omega(t)}.
\else
\Pr[Z_P > t d(P) / \lambda] 
  \le \Pr\Big[\exists i \in [r] \ \text{such that} \ Z_{P_i} > t/2\Big] 
  \le O \left(\left\lceil \frac{d(P)}{\lambda} \right\rceil \right) \cdot e^{- \Omega(t)}.
\fi
\end{align*}
Furthermore, for every $j \in [k]$, 
let ${\cal A}_j \eqdef \{i \in [r] : P_i \cap B(t_j,\Delta) \ne \emptyset\}$, 
and since $P$ is a shortest path, 
$|{\cal A}_j| \le 4 \Delta / \lambda = 4 \log k$.
Observe that $Z_{P_i}=0$ (with certainty) 
for all $i\notin \cup_{j}{\cal A}_j$,
hence
\begin{align*}
\ifprocs
\Pr[Z_{P} &> t d(P) / \lambda ] \\
  &\le \Pr\Big[\exists i \in \cup_{j\in[k]}{\cal A}_j\ \text{ such that } Z_{P_i} > t/2 \Big] \\
  &\le 4k \log k e^{- \Omega(t)} \; .
\else
\Pr[Z_{P} > t d(P) / \lambda ] 
  \le \Pr\Big[\exists i \in \cup_{j\in[k]}{\cal A}_j\ \text{ such that } Z_{P_i} > t/2 \Big] 
  \le 4k \log k e^{- \Omega(t)} \; .
\fi
\qedhere
\end{align*}
\end{proof}

By substituting $\beta = 4 \log k$ and $\lambda = \Delta/ \log k$, it is easy to verify that Lemmas \ref{l:shortShortest} and \ref{l:longShortest}
complete the proof of Theorem \ref{thm:main3}.

\section{Terminal-Centered Minors: Main Construction} \label{sec:TCM}

This section proves Theorem \ref{thm:main} when 
$\calD \eqdef \frac{\max_{u,v \in T}{d_G(u,v)}}{\min_{u,v \in T}{d_G(u,v)}}$
satisfies the following assumption 
(the extension to the general case is proved in Section~\ref{sec:proofDiscard}).
\begin{assumption} \label{a:polyLogDelta}
$\calD \le 2^{k^3}$.
\end{assumption}
By scaling all edge weights, we may further assume that 
$\min_{u,v \in T}{d_G(u,v)} = 1$.

\begin{notation}
Let $V_1, \ldots, V_k \subseteq V$.
For $S \subseteq [k]$, denote $V_S \eqdef \bigcup_{j \in S}V_j$. 
In addition, denote $V_\bot \eqdef V \setminus V_{[k]}$ 
and $V_{\bot+j} \eqdef V_\bot \cup V_j$ for any $j \in [k]$. 
\end{notation}

We now present a randomized algorithm that, 
given a graph $G=(V,E,w)$ and terminals $T\subset V$,
constructs a terminal-centered minor $G'$ as stated in Theorem~\ref{thm:main}.
The algorithm maintains a partial partition $\{V_1,V_2,\ldots,V_k\}$ of $V$, 
starting with $V_j = \{t_j\}$ for all $j \in [k]$. 
The sets grow monotonically during the execution of the algorithm. 
We may also think of the algorithm as if it maintains a mapping 
$f : V \to T \cup \{\bot\}$, 
starting with $f(t_j)=t_j$ for all $j \in [k]$ and gradually assigning a value in $T$ 
to additional vertices, which correspond to the set $V_{[k]}$. 
Thus, we will also refer to the vertices in $V_{[k]}$ as {\em assigned}, 
and to vertices in $V_\bot$ as {\em unassigned}.
The heart of the algorithm is two nested loops (lines \ref{al:outer}-\ref{al:endOuter}). During every iteration of the outer loop, the inner loop performs $k$ iterations, one for every terminal $t_j$. 
Every inner-loop iteration picks a random radius (from an exponential distribution) 
and ``grows'' $V_j$ to that radius (but without overlapping any other set) thus removing nodes from $V_\bot$ and assigning them to $V_j$. 
Every outer-loop iteration increases the expectation of the radius distribution.
Eventually, all nodes are assigned, 
i.e.\ $\{V_1,V_2,\ldots,V_k\}$ is a partition of $V$.
Note that the algorithm does not actually contract the clusters at the end of each iteration of the outer loop. 
However, subsequent iterations grow each $V_j$ 
only in the subgraph of $G$ induced by the respective $V_{\bot+j}$,
which is effectively the same as contracting clusters to their respective terminals at the end of each outer-loop iteration.

Every cluster in the partial partition maintained by the algorithm needs to induce a connected subgraph of $G$, 
and thus we cannot directly use the result of Section~\ref{sec:main2}, 
but rather apply more subtle arguments which use the same idea. 
In particular, the algorithm has to grow the clusters so that they do not overlap. We therefore require the following definition.
\begin{definition}
For $U \subseteq V$, let $G[U]$ denote the subgraph of $G$ induced by $U$, with induced edge lengths (i.e. $w|_{E(G[U])}$).
For a subgraph $H$ of $G$ with induced edge lengths, a vertex $v \in V(H)$ and $r>0$, denote $B_H(v,r) \eqdef \{u \in V(H):\ d_H(u,v) \le r\}$, where $d_H$ is the shortest path metric in $H$ induced by $w$.
\end{definition}

\begin{algorithm}[H]
\caption{Partitioning $V$}
\label{alg:part}
\begin{algorithmic}[1]
\REQUIRE $G = (V,E,w),\; T =\{t_1,\ldots,t_k\}\subseteq V$
\ENSURE A partition $\{V_1,V_2,\ldots,V_k\}$ of $V$.
\STATE set $b \leftarrow 1 + 1/(45\log k)$
\STATE for every $j \in [k]$ set $V_j \leftarrow \{t_j\}$, $r_j=0$.
\STATE set $i \leftarrow 0$. \COMMENT{$i$ is the iteration number of the outer loop.}
\WHILE{$V_{[k]} \ne V$} \label{al:outer} 
\STATE $i \leftarrow i+1$.
\FORALL{$j \in [k]$} \label{al:inner}
\ifprocs
\STATE choose independently\! at\!  random\! $R_j^i \sim \exp(b^i)$.
\else
\STATE choose independently at  random $R_j^i \sim \exp(b^i)$.
\fi
\STATE $r_j \leftarrow r_j + R_j^i$.
\STATE $V_j \leftarrow V_j \cup B_{G[V_{\bot+j}]}(t_j,r_j)$. 
  \COMMENT{This is the same as $V_j \leftarrow B_{G[V_{\bot+j}]}(t_j,r_j)$.}
\ENDFOR
\ENDWHILE \label{al:endOuter}
\RETURN $\{V_1,V_2,\ldots,V_k\}$.
\end{algorithmic}
\end{algorithm}

\begin{claim}
The following properties hold throughout the execution of the algorithm.
\begin{enumerate} \compactify
	\item For all $j \in [k]$, $V_j$ is connected in $G$, and $t_j \in V_j$.
	\item For every $j_1,j_2 \in [k]$, if $j_1 \ne j_2$, then $V_{j_1} \cap V_{j_2} = \emptyset$.
	\item For every outer loop iteration $i$ and every $j \in [k]$, 
if $V_j'$ denotes the set $V_j$ at the beginning of the $i$-th iteration (of the outer loop), and $V_j''$ denotes the set $V_j$ at the end of that iteration, 
then $V_j' \subseteq V_j''$.
\end{enumerate}
\end{claim}

In what follows, we analyze the stretch in distance between a fixed pair of terminals. We show that with probability at least $1 - O(k^{-5})$, the distance between these terminals in $G'$ is at most $O(\log ^5 k)$ times their distance in $G$. By a union bound over all $\binom{k}{2}$ pairs of terminals, we deduce Theorem~\ref{thm:main}.
Let $s,t \in T$, and let $P^*$ be a shortest $st$-path in $G$. Due to the triangle inequality, we may focus on pairs which satisfy $V(P^*) \cap T = \{s,t\}$, where $V(P^*)$ is the node set of $P^*$. We denote $\ell \eqdef w(P^*) = d_{G,w}(s,t)$.

\subsection{High-Level Analysis}
Following an execution of the algorithm, we maintain a (dynamic) path $P$ between $s$ and $t$. In a sense, in every step of the algorithm, $P$ simulates an $st$-path in the terminal-centered minor induced by $V_1,V_2,\ldots,V_k$. At the beginning of the execution, set $P$ to be simply $P^*$. During the course of the execution update $P$ to satisfy two invariants. At every step of the algorithm, the weight of $P$ is an upper bound on the distance between $s$ and $t$ in the terminal centered minor induced by $V_1,\ldots, V_k$ (in that step). In addition, if $I$ is a subpath of $P$, whose inner vertices are all unassigned, then $I$ is a subpath of $P^*$. Throughout the analysis, we think of $P$ as directed from $s$ to $t$, thus inducing a linear ordering of the vertices in $P$. 

\begin{definition}
A subpath of $P$ will be called {\em active} if it is a maximal subpath whose inner vertices are unassigned.
\end{definition}

Note that a single edge whose endpoints are both assigned will not be considered active.

We now describe how $P$ is updated during the execution of the algorithm. Consider line \ref{al:endOuter} of the algorithm for the $i$-th iteration of the outer loop, and some $j \in [k]$. We say that the ball $B=B_{G[V_{\bot+j}]}(t_j,r_j)$ {\em punctures} an active subpath $A$ of $P$, if there is an inner node of $A$ that belongs to the ball.
If $B$ does not puncture any active subpath of $P$, we do not change $P$. Otherwise, denote by $u,v$ the first and last unassigned nodes (possibly not in the same active subpath) in $V(P) \cap B$ respectively. Then we do the following.

We replace the entire subpath of $P$ between $u$ and $v$ with a concatenation of a shortest $ut_j$-path and a shortest $t_jv$-path that lie in $B$; this is possible, since $G[B]$ is connected, and $u,t_j,v \in B$.
This addition to $P$ will be called a {\em detour} from $u$ to $v$ through $t_j$. The process is illustrated in figures~\ref{f:det-a}-\ref{f:det-b}. Beginning with $P^*$, the figure describes the update after the first four balls. Note that the detour might not be a simple path.
It is also worth noting that here $u$ and $v$ may belong to different active subpaths of $P$. For example, in figure~\ref{f:det-c}, the new ball punctures two active subpaths, and therefore in figure~\ref{f:det-d}, the detour goes from a node in one active subpath to a node in another active subpath. Note that in this case, we remove from $P$ portions which are not active.

It is worth noting that this update process implies that at any given time, there is at most one detour that goes through $t_j$. If, for some iteration $i'<i$ of the outer loop, and for some $u',v' \in V(P)$, we added a detour from $u'$ to $v'$ through $t_j$ in iteration $i'$, we keep only one detour through $t_j$, from the first node between $u,u'$ and to the last between $v,v'$. For example, in figure~\ref{f:det-e}, the ball centered in $t_3$ punctures an active subpath. Only one detour is kept in figure~\ref{f:det-f}.
\ifprocs 
\begin{figure*}[ht]
  \begin{center}
    \subfigure[We begin with $P=P^*$]{\label{f:det-a}\includegraphics[scale=0.25]{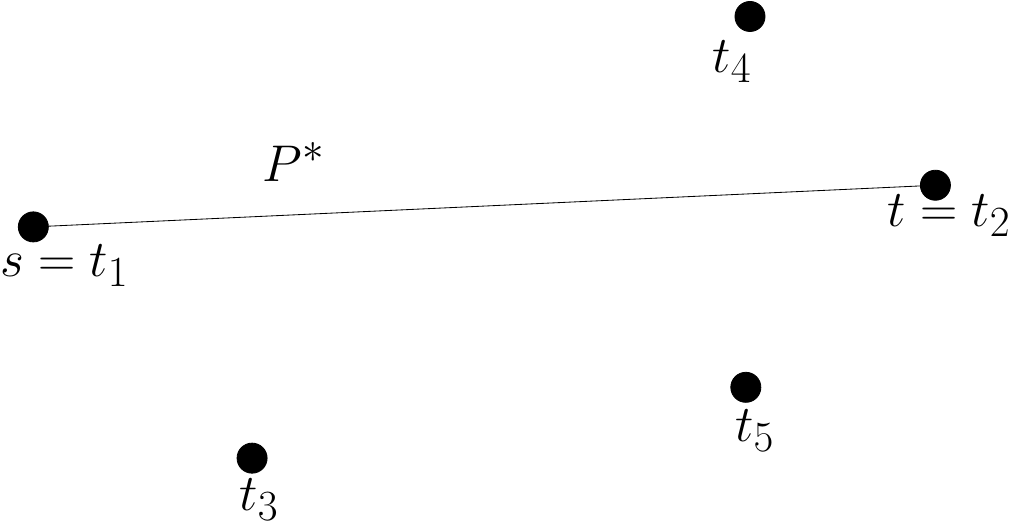}}
    \ 
    \subfigure[Every inner-loop iteration grows a terminal-centered ball. Here balls around $t_1, t_2, t_3, t_4$ are grown with detours added. Since $P^*$ is a shortest path, the detours for $t_1$ and $t_2$ are, in fact, subpaths of $P^*$.]{\label{f:det-b}\includegraphics[scale=0.25]{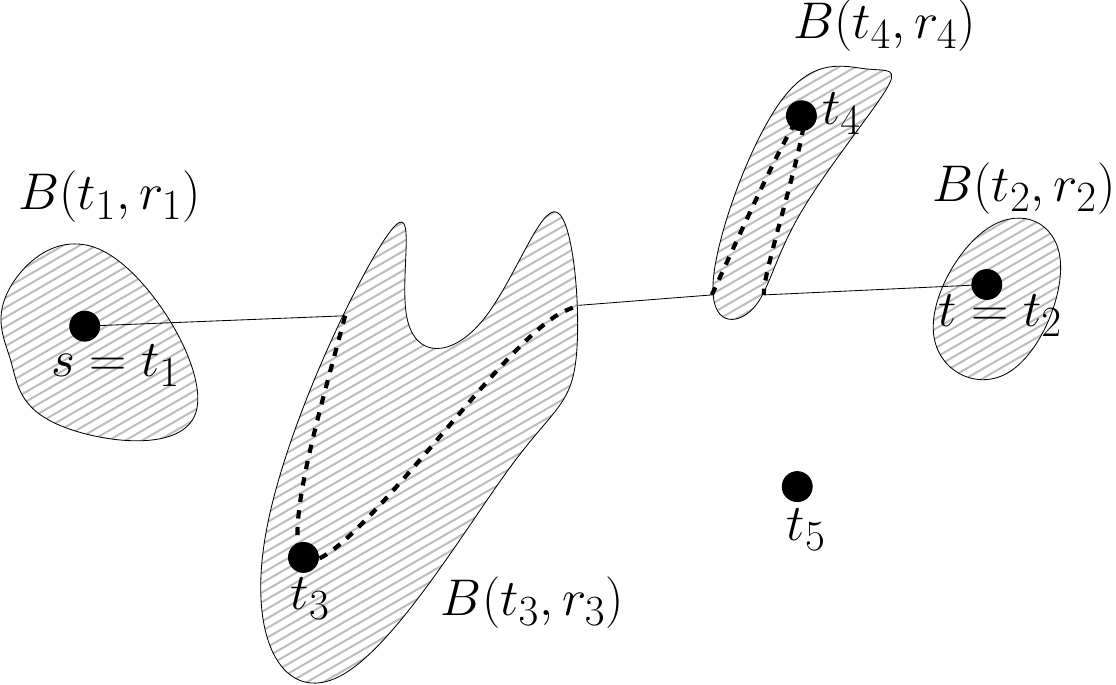}}
    \ 
    \subfigure[Endpoints of a detour can belong to different active subpaths. Here a ball around $t_5$ is grown.]{\label{f:det-c}\includegraphics[scale=0.25]{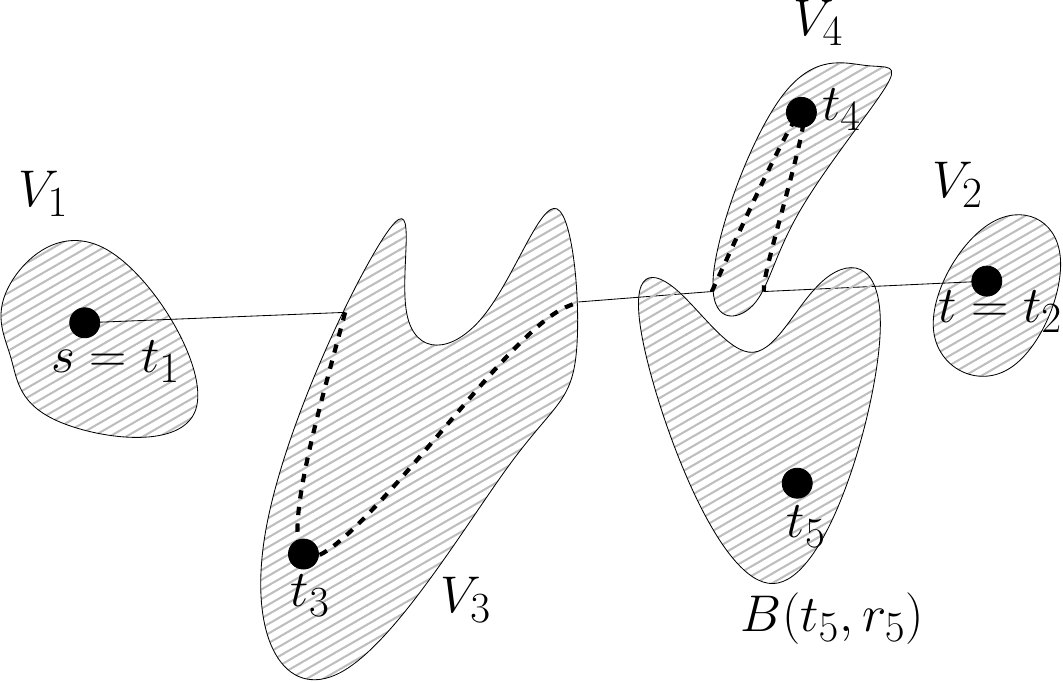}}
    \
    \subfigure[Update detour for $V_5$.]{\label{f:det-d}\includegraphics[scale=0.25]{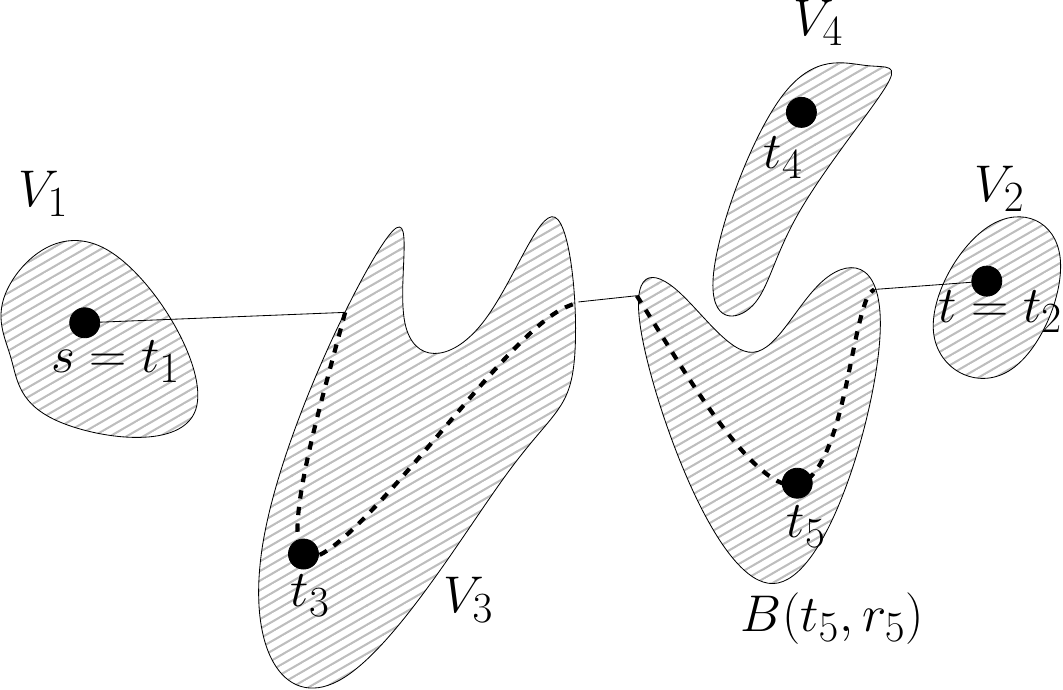}}
    \ 
    \subfigure[The ball around $t_3$ is further grown.]{\label{f:det-e}\includegraphics[scale=0.25]{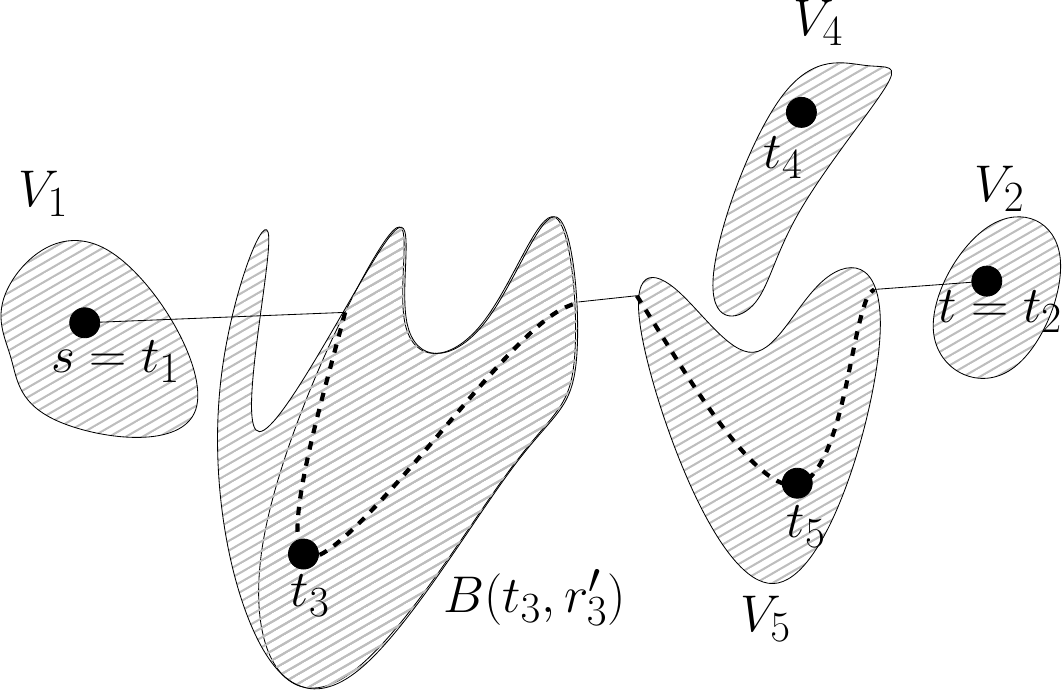}}
    \ 
    \subfigure[Update detour for $V_3$.]{\label{f:det-f}\includegraphics[scale=0.25]{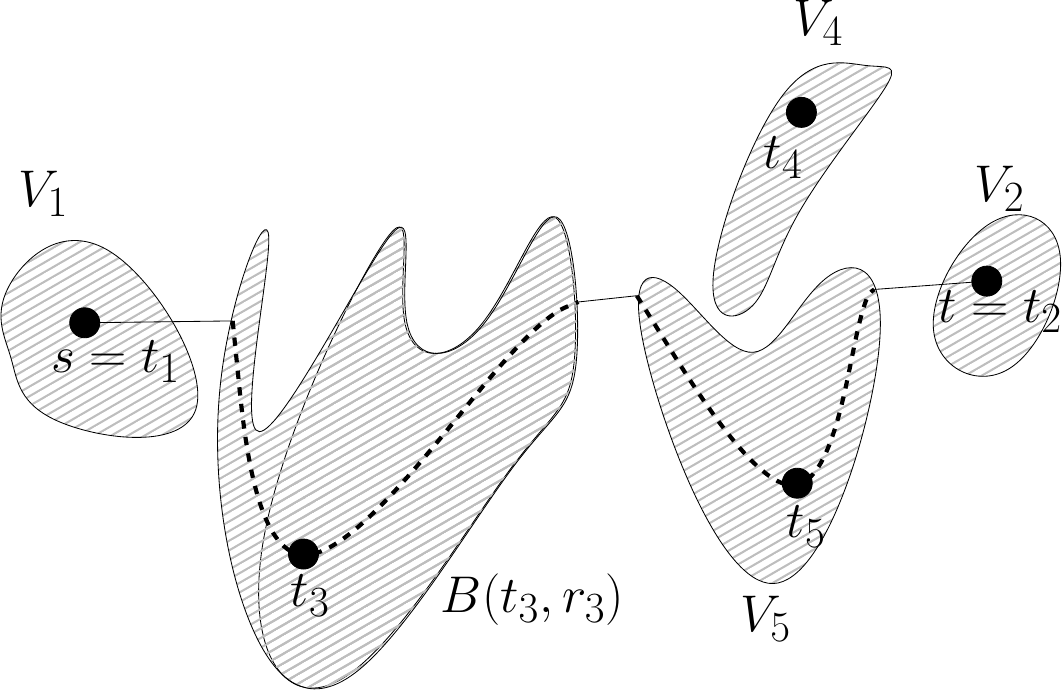}}
  \end{center}
  \caption{Updating $P$}
  \label{f:det}
\end{figure*}
\else
\begin{figure}[t]
  \begin{center}
    \subfigure[We begin with $P=P^*$]{\label{f:det-a}\includegraphics[scale=0.25]{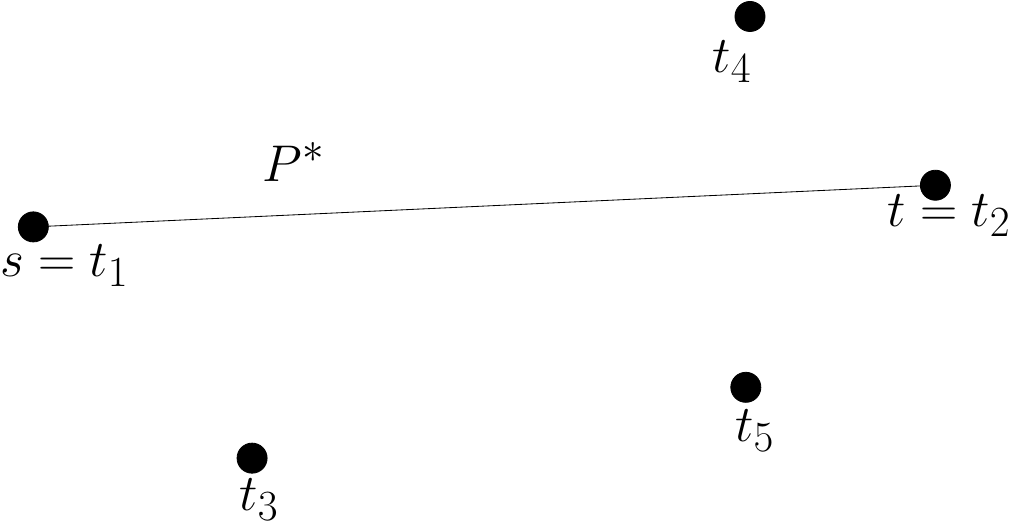}}
    \quad 
    \subfigure[Every inner-loop iteration grows a terminal-centered ball. Here balls around $t_1, t_2, t_3, t_4$ are grown with detours added. Since $P^*$ is a shortest path, the detours for $t_1$ and $t_2$ are, in fact, subpaths of $P^*$.]{\label{f:det-b}\includegraphics[scale=0.25]{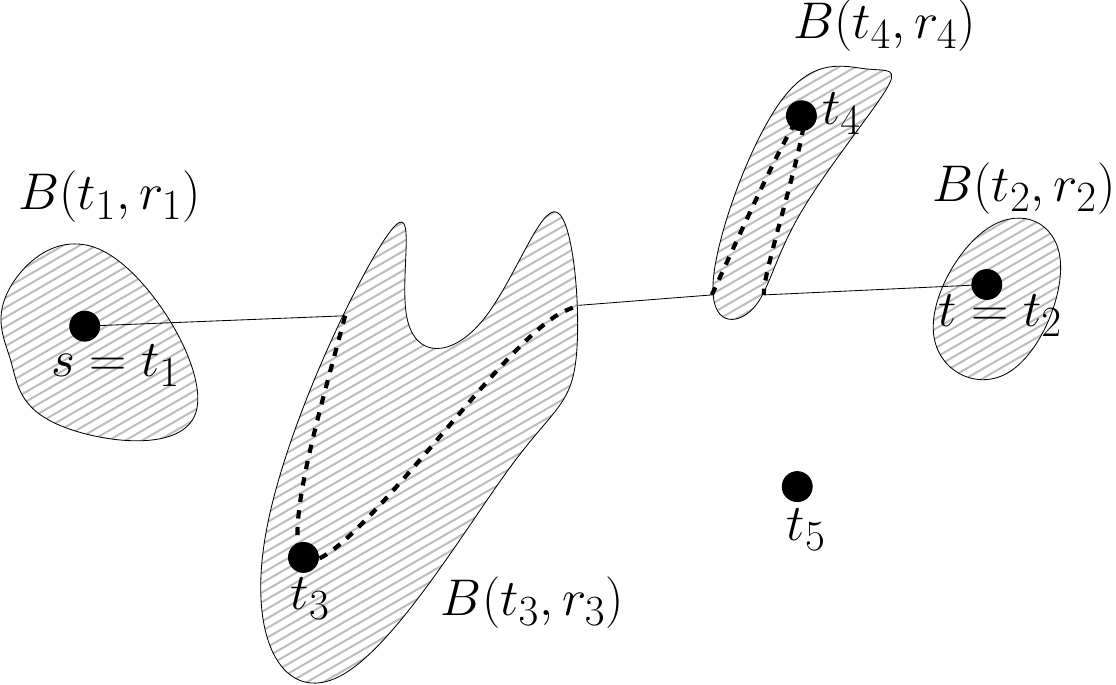}}
    \quad
    \subfigure[Endpoints of a detour can belong to different active subpaths. Here a ball around $t_5$ is grown.]{\label{f:det-c}\includegraphics[scale=0.25]{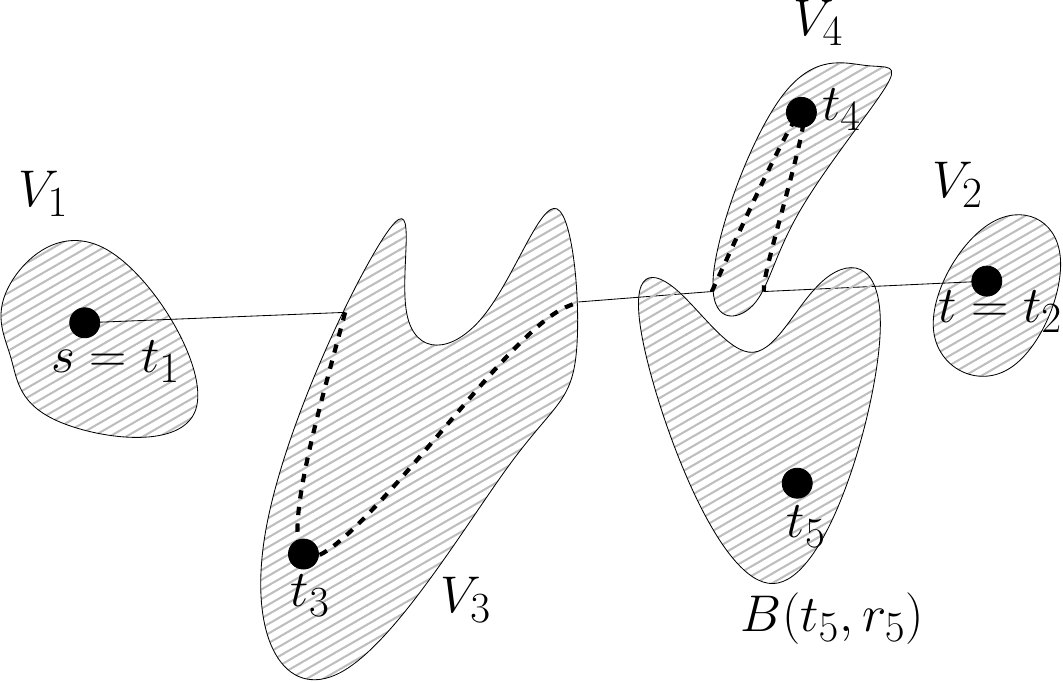}}
    \quad
    \subfigure[Update detour for $V_5$.]{\label{f:det-d}\includegraphics[scale=0.25]{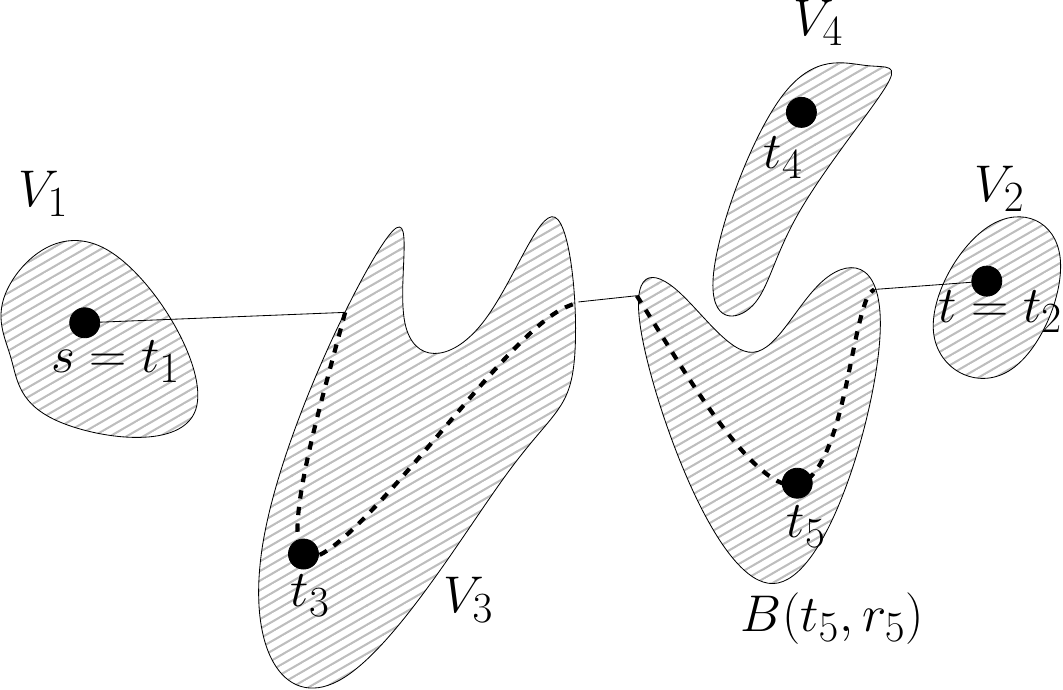}}
    \quad
    \subfigure[The ball around $t_3$ is further grown.]{\label{f:det-e}\includegraphics[scale=0.25]{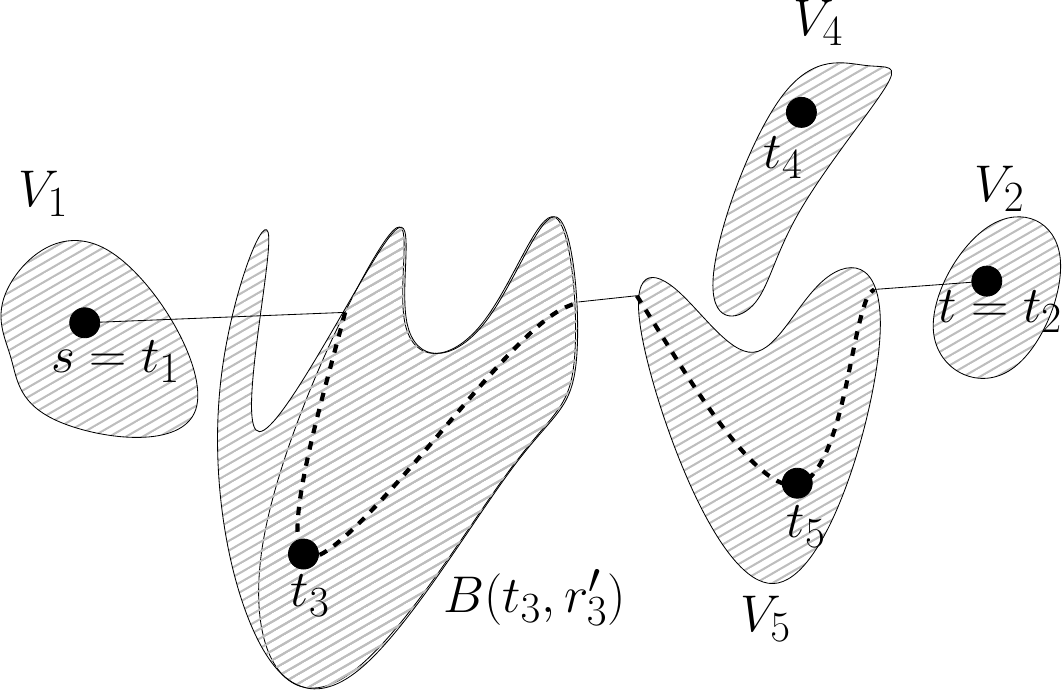}}
    \quad
    \subfigure[Update detour for $V_3$.]{\label{f:det-f}\includegraphics[scale=0.25]{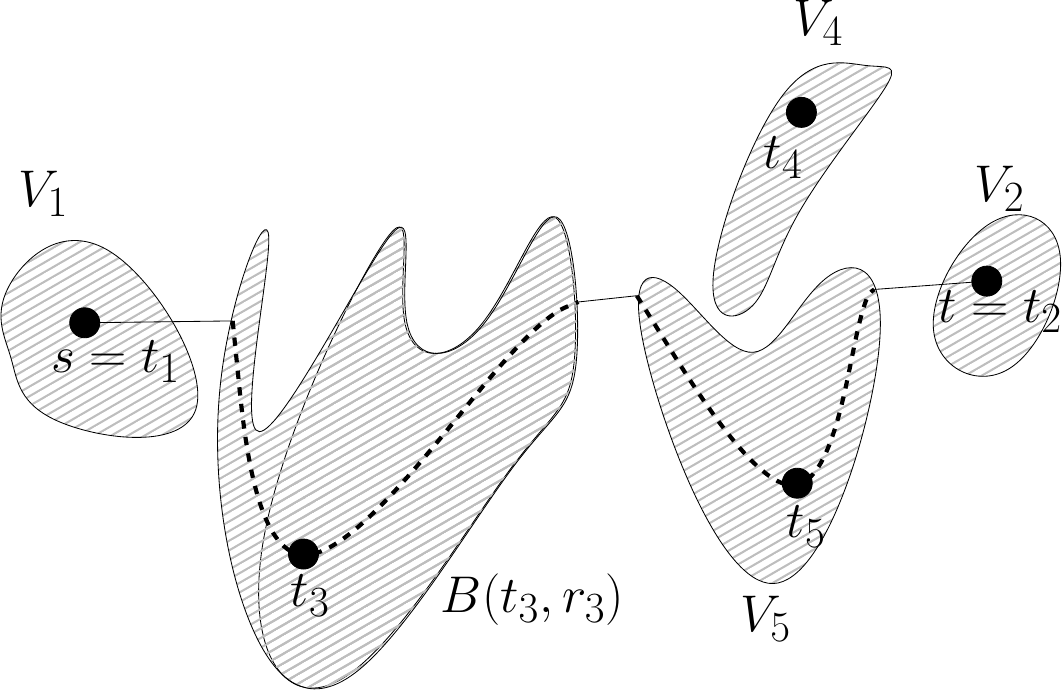}}
  \caption{Updating $P$}
  \label{f:det}
  \end{center}
\hrule
\end{figure}
\fi

The total weight of all detours during the execution will be called the {\em additional weight} to $P$ (ignoring portions of $P$ that are deleted from $P$).
Denote the set of active subpaths of $P$ at the beginning of the $i$-th iteration of the outer loop by ${\cal A}_i$.

Let $V_1^\fin,\ldots,V_k^\fin$ be the partition returned by the algorithm, let $G'$ the terminal-centered minor induced by that partition, and let $w'$ be the standard restriction of $w$ to $G'$. 
Denote by $P^\fin$ the path obtained at the end of the execution.
\begin{claim} 
At every step of the algorithm the following holds:
\begin{enumerate}
	\item The weight of $P$ is an upper bound on the distance between $s$ and $t$ in the terminal centered minor induced by $V_1,\ldots, V_k$. Moreover, once ${\cal A}_i = \emptyset$ (namely, $P$ has no active subpaths), the weight of $P$ is an upper bound on the distance between $s$ and $t$ in the terminal centered minor induced by $V_1^\fin,\ldots, V_k^\fin$ (actually, from this point on, $P=P^\fin$).
	\item If $A$ is a subpath of $P$, whose inner points are all in $V_\bot$, then $A$ is a subpath of $P^*$.
	\item If $A_1,A_2$ are two different active subpaths of $P$, they are internally disjoint.
	\item $|{\cal A}_i| \le k$ for all $i$.
\end{enumerate}
\end{claim}

\begin{proof}
Follows easily by induction on $i,j$.
\end{proof} 
\begin{corollary} \label{cor:costMinor}
$d_{G',w'}(s,t) \le w(P^\fin)$.
\end{corollary}

Let $A \in {\cal A}_i$. During the execution of the inner loop, $A$ is either removed from $P$ entirely, or some subpaths of $A$ remain active (perhaps $A$ remains active entirely). Therefore, for every $A' \in {\cal A}_{i+1}$, either $A'$ is a non-trivial subpath of $A$ (by non-trivial we mean $|V(A')| \ge 3$), or $A'$ and $A$ are internally disjoint. Therefore there is a laminar structure on $\bigcup_i{\cal A}_i$. 

We describe this structure using a tree ${\cal T}$, whose node set is $\{\langle i, A \rangle | \ A \in {\cal A}_i\}$. The root of ${\cal T}$ is $\langle 1, P^* \rangle$, and for every $i$ and every $A \in {\cal A}_i$, the children of $\langle i,A \rangle$, if any, are all pairs $\langle i+1 ,A'\rangle$, where $A' \in {\cal A}_{i+1}$ is a subpath of $A$.
Whenever we update $P$ we log the weight of the detour by charging it to one of the nodes of ${\cal T}$ as follows. Consider a detour from $u$ to $v$ in the $i$-th iteration of the outer loop for some $i$. Before adding this detour, $u$ and $v$ are unassigned nodes in $P$. Because $u$ is unassigned, $u$ is an inner vertex of some active subpath. In either case, there is exactly one active subpath containing $u$. The weight of the detour is charged to the unique active subpath $A \in {\cal A}_i$ such that $u \in A$. For every $i$ and $A \in {\cal A}_i$, let $w_{i,A}$ be the total weight charged to $\langle i,A \rangle$. If the node is never charged, the weight of the node is set to $0$. Therefore, 
\begin{equation}
w(P^\fin) \le w(P^*) + \sum_{\langle i,A \rangle \in {\cal T}}{w_{i,A}} .
\label{eq:treeBound}
\end{equation}
Eqn.~\eqref{eq:treeBound} together with Corollary~\ref{cor:costMinor} imply that if we show that with high probability, the total weight charged to the tree is at most $O(\log ^5k)\cdot \ell$, we can deduce Theorem~\ref{thm:main}.
For the rest of this section, we therefore prove the following lemma.
\begin{lemma} \label{l:totalWeight}
With probability at least $1 - O(k^{-5})$, the total weight charged to the tree is at most $O(\log ^5 k) \ell$.
\end{lemma}

Consider an iteration $i \ge 1$ and an active subpath $A \in {\cal A}_i$. 
Informally, since the distortion is measured relatively to $\ell = w(P^*)$, 
if the expected radius $b^i$ is small compared to $w(A)$, 
then with high probability a detour will not add ``much'' to the distortion,
and thus we are more concerned with the opposite case 
where $w(A)$ is small relative to the current expected radius.

Formally, let $p = 1/100$. An active subpath $A \in {\cal A}_i$ will be called {\em short} if $w(A) \le pb^i$. Otherwise, $A$ will be called {\em long}. Notice that $P^* \in {\cal A}_1$ is long, and for $i \ge \log_b(\ell/p)$, every $A \in {\cal A}_i$ is short.

\begin{definition}
Let $i>1$ and let $A \in {\cal A}_i$ be a short subpath. Denote by ${\cal T}_{i,A}$ the subtree of ${\cal T}$ rooted in $\langle i,A \rangle$. 
Denote the parent of $\langle i,A \rangle$ in ${\cal T}$ 
by $\langle i-1,A' \rangle$ for $A' \in {\cal A}_{i-1}$.
If $A'$ is long, ${\cal T}_{i,A}$ will be called a {\em short subtree} of ${\cal T}$.
\end{definition}

Once an active subpath becomes short (during the course of the iterations), 
we want all its vertices to be assigned quickly, and by a few detours. 
For this reason, the height and weight of short subtrees will play an important role in the analysis of the height and weight of ${\cal T}$.

To bound the total weight of the tree ${\cal T}$, we analyze separately 
the weights charged to long active subpaths at each level, 
and the weights of short subtrees rooted at each level.
More formally, for every $i \le \log_b(\ell/p)$, denote by $l_i$ the total weight charged to nodes of the form $\left\langle i,A \right\rangle$, where $A \in {\cal A}_i$ is a long active subpath. Denote by $s_i$ the total weight charged to short subtrees rooted at the level $i$ of ${\cal T}$. 
For $i \ge \log_b(\ell/p)$, every $A \in {\cal A}_i$ is short and thus $\left\langle i,A \right\rangle$ belongs to some short subtree rooted at level at most $\log_b (\ell/p)$. Therefore,
\begin{equation}
\sum_{\langle i,A \rangle \in {\cal T}}{w_{i,A}} = \sum_{i=1}^{\log_b (\ell/p)}{(l_i+s_i)} \;.
\label{eq:treeWeight}
\end{equation}
Similarly to the proof of Theorem~\ref{thm:main3}, we will first analyze the behavior of short subpaths of active paths, and then use it to bound $s_i$. 
To bound the weight of long active paths, we will divide them into short segments, similarly to the proof of Lemma~\ref{l:longShortest}, and then sum everything up to bound $l_i$.

\subsubsection{The Effect of a Single Ball on a Short Segment}
Let $i_0 \ge 1$ and let $I$ be a subpath of $P$ such that all the inner nodes of $I$ are unassigned in the beginning of the $i_0$-th iteration of the outer loop, and $w(I) \le pb^{i_0}$. Note that $I$ is not necessarily maximal with that property, and therefore is not necessarily an active subpath. However, $I$ is a subpath of some (unique) active subpath $A \in {\cal A}_{i_0}$.
We first consider the effect of a single ball over $I$, in some iteration $i \ge i_0$. 

Fix some $i \ge i_0$, and some $j \in [k]$.
Let $X$ denote the number of active subpaths $A'$ such that $V(A') \cap V(I) \ne \emptyset$ at the beginning of the $j$-th iteration of the inner loop (during the $i$-th iteration of the outer loop). Note that every such subpath $A'$ is necessarily a subpath of $A$, due to the laminar structure of active subpaths. Since every such active subpath will add at least one detour before it is completely assigned, we want to show that $X$ is rapidly decreasing.
Let $X'$ denote the number of active subpaths $A'$ such that $V(A') \cap V(I) \ne \emptyset$ at the end of the $j$-th iteration. Denote by $B$ the ball considered in this iteration, namely $B \eqdef B_{G[V_{\bot+j}]}(t_j,r_j)$.
\begin{proposition} \label{p:maxBirth}
With certainty, $X' \le X+1$.
\end{proposition} 

\begin{proof}
Let $A_1,A_2,\ldots,A_X$ be all active subpaths of $A$ which intersect $I$ and are active in the beginning of the $j$-th iteration ordered by their location on $P$. For $\alpha \in [X]$ denote by $u_\alpha,v_\alpha$ the first and last unassigned nodes in $A_\alpha$, respectively. If $B$ does not puncture any of these subpaths, then $X' \le X < X+1$ (Note that subpaths of $A$ can still be removed if $B$ punctures active subpaths of $P$ not contained in $A$). So assume $B$ punctures $A_\alpha$.
Assume first that $A_\alpha$ is the only subpath of $P$ which is active and is punctured by $B$. Then there are three options:
If both $u_\alpha,v_\alpha \in B$, then $A_\alpha$ is replaced and removed entirely from $P$ when adding the detour, and $X' \le X-1 < X+1$.
If $u_\alpha \in B$ and $v_\alpha \notin B$, let $v'$ be the last node in $V(A_\alpha) \cap B$ then the $u_\alpha v'$ segment of $A_\alpha$ is replaced, and the segment $v'v_\alpha$ remains active. Therefore $X' \le X < X+1$. The argument is similar, if $u_\alpha \notin B$ and $v_\alpha \in B$.
Otherwise, some of the inner portion of $A_\alpha$ is replaced by a non-active path, and both end segments of $A_\alpha$ remain active, therefore $X' = X+1$.
Next, assume the ball punctures several active subpaths of $A$, and maybe more subpaths of $P$. Denote by $I_\alpha, I_\beta$ the first and last subpaths of $A$ punctured by $B$. Denote by $u$ the first node in $V(I_\alpha) \cap B$, and $v$ the last node in $V(I_\beta) \cap B$. When updating $P$, the entire subpath of $P$ between $u$ and $v$ is removed. 
Thus $X' \le X -(\beta - \alpha + 1) \le X < X+1$.
\end{proof}

We now want to show that if some unassigned vertex $v \in V(I)$ gets assigned due to a ball $B$, i.e., $B$ punctures some active subpath intersecting $I$, then $X$ is likely to decrease. 
Recalling that unassigned nodes in $P$ must be in $P^*$, 
this goal is stated formally as 
$$
  \Pr[X' \ge X \mid B \cap V(I) \cap V(P^*) \ne \emptyset] \le p \;.
$$ 
However, this statement is not sufficient for our needs, as it does not imply 
that with high probability a short active subpath is assigned quickly. 
Indeed, let $I$ be a short active subpath and suppose no ball punctures any subpath of $I$ for many iterations following $i_0$; 
then the detour that will eventually be added to replace a subpath of $I$
might be too long relative to $I$ (as expected radii increase exponentially). %
Therefore, when arguing that with reasonable probability $X$ decreases,
we shall condition on a more refined event, 
which generalizes the notion of a ball puncturing an active subpath.
Loosely speaking, we consider events in which the ball $B$ includes 
a vertex $v \notin P^*$ (i.e., $v$ is already assigned),
and assume there is an unassigned $u \in V(I)$, such that $uv$ is an edge in $G$ (since $P^*$ is a shortest path, this edge $uv$ must be part of $P^*$),
which means there is an active subpath intersecting $I$ adjacent to $v$
(in particular, $v$ is one of its endpoints). 
By the memoryless property of the exponential distribution, conditioned on $v \in B$, with reasonable probability $B$ covers that subpath. 
The formal definition follows.

\begin{definition}
Let $I$ be a subpath of $P$. We say that a ball $B$ {\em reaches} $I$ if there is $v \in V(I) \cap B$ such that either $v \in V(P^*)$ is unassigned, 
or $v$ has an unassigned neighbor which is in $V(I)$.
\end{definition}

Consider again the case where $I$ is active. Then both its endpoints are assigned. Note that the endpoints of $I$ cannot both be assigned to the same terminal (otherwise $I$ would have been removed entirely). 
By the definition of the balls in the algorithm, $B \subseteq G[V_{\bot+j}]$ and therefore $B$ may reach $I$ and not puncture it if and only if $I$ has exactly one endpoint in $V_j$. Note that all active subpaths are reached at least twice in every iteration of the outer loop (by the clusters which contain their endpoints). Therefore in every iteration of the outer loop at least two balls reach $I$ with certainty, even though it could be the case that no ball punctures $I$. 
\begin{proposition} \label{p:cond}
$\Pr[X' \ge X \mid \text{$B$ reaches $I$}] \le p \;.$
\end{proposition} 

\begin{proof}
Assume that $B$ reaches $I$. Then there exists a node $v \in V(I) \cap B$ such that either $v$ is unassigned, or $v$ has an unassigned neighbor $u \in V(I)$.
Let $d = d_{G[V_{\bot+j}]}(t_j,v)$.
Assume first that $v \in V(P^*)$ is unassigned. Let $A'$ be the active subpath such that $v \in V(A')$.
Following the analysis of the previous proof, if $X' \ge X$, then $B$ punctures exactly one active subpath (namely $A'$) that intersects $I$ and does not cover the part of $A'$ contained in $I$, or $B$ punctures exactly two  such active subpaths and covers neither of them. In either case, $B$ punctures $A'$ and does not cover the part of $A'$ contained in $I$. Since $w(I) \le pb^{i_0} \le pb^i$, the length of $A'$ is at most $pb^i$. We conclude that $r_j + R^i_j \ge d$, and $r_j + R^i_j < d + pb^i$.
If $v$ has an unassigned neighbor $u \in V(I)$, we get the same conclusion, since this again means $r_j +R_j^i \ge r_j \ge d$. By the memoryless property,
\ifprocs
\begin{align*}
\Pr[X' \ge X &\mid \text{$B$ reaches $I$}] 
\\ &\le \Pr[R_{j}^i < d-r_{j} + pb^i \mid R_{j}^i \ge d-r_{j}] \\ &\le  1-e^{-p} \le p \;.
\end{align*}
\else
\begin{align*}
\Pr[X' \ge X \mid \text{$B$ reaches $I$}] 
\le \Pr[R_{j}^i < d-r_{j} + pb^i \mid R_{j}^i \ge d-r_{j}] \le  1-e^{-p} \le p \;.
\end{align*}
\fi
\end{proof}

\subsubsection{The Effect of a Sequence of Balls on a Short Segment}
Consider now the first $N$ balls that reach $I$, starting from the beginning of iteration $i_0$ of the outer loop, and perhaps during several iterations of that loop. For every $a \in [N]$, let $Y_a$ be the indicator random variable for the event that the $a$-th ball reaching $I$ decreased the number of active subpaths intersecting $I$. In these notations, Proposition~\ref{p:cond} stated that
$$
  \forall a \in [N],\qquad 
  \Pr[Y_{a+1}=1 \mid Y_1,\ldots,Y_a] \ge 1 - p \;.
$$
Let $Y = \sum_{a \in [N]}{Y_a}$ and let $Z \sim Bin(N,1-p)$. Simple induction on $N$ implies the following claim.
\begin{claim}
$\forall k, \; \Pr\left[Y > k\right] \ge \Pr[Z>k]$.
\end{claim}
\begin{lemma} \label{l:puncBalls}
With probability at least $1-1/k^{10}$, after $90 \log k$ balls have reached $I$, there are no active subpaths intersecting $I$.
\end{lemma} 
\begin{proof}
Assume $N = 70 \log k$.
Since whenever $Y_a=0$, the number of active subpaths increases by at most $1$, and whenever $Y_a=1$, the number of active subpaths decreases by at least $1$, if $Y > N/2$, then there are no active subpaths intersecting $I$. Therefore by the Chernoff bound,
\begin{align*}
\ifprocs
\Pr[\text{there are no active subpaths}&\text{ intersecting $I$ after $N$ balls reach $I$}] \\
\ge \Pr[Y > N/2] \ge \Pr[Z>N/2] \ge 1 - 1/k^{10} \;.&
\else
\Pr[\text{there are no active subpaths}&\text{ intersecting $I$ after $N$ balls reach $I$}] \\
&\ge \Pr[Y > N/2] \ge \Pr[Z>N/2] \ge 1 - 1/k^{10} \;.
\fi
\end{align*}
\end{proof}

\subsection{The Behavior of Short Subtrees} \label{subsec:short}
As stated before, the most crucial part of the proof is to bound the weight and height of short subtrees of ${\cal T}$.
Let $i_0 > 1$, and let $A \in {\cal A}_{i_0}$ be a short subpath such that ${\cal T}_{i_0,A}$ is a short subtree of ${\cal T}$.
Clearly, $A'$ is short for every node $\langle i',A' \rangle$ of ${\cal T}_{i_0,A}$. 
In order to bound the height of ${\cal T}_{i_0,A}$ we combine the fact that not too many balls may reach $A$, with the fact that at least two balls reach $A$ during each iteration of the outer loop. 
\begin{claim}
With probability at least $1-1/k^{10}$, the height of ${\cal T}_{i_0,A}$ is at most $45\log k$.
\end{claim}
\begin{proof}
In the notations of Lemma~\ref{l:puncBalls}, consider some $i \ge i_0$.
If $A$ has an active subpath at the end of the $i$-th iteration of the outer loop, then at least two times during the $i$-th iteration an active subpath of $A$ is reached by a ball. 
After $45\log k$ iterations of the outer loop, if $A$ has an active subpath, then $N \ge 90 \log k$.
By similar arguments to Lemma~\ref{l:puncBalls}, 
\begin{align*}
\ifprocs
\Pr[\text{The height of ${\cal T}_{i_0,A}$ is at most $N/2$}] &\ge \Pr[Y > N/2] \\ &\ge \Pr[Z>N/2] \\ &\ge 1 - 1/k^{10} \;.
\else
\Pr[\text{The height of ${\cal T}_{i_0,A}$ is at most $N/2$}] \ge \Pr[Y > N/2] \ge \Pr[Z>N/2] \ge 1 - 1/k^{10} \;.
\fi
\end{align*}
\end{proof}
We denote by $\mathbf{E_1}$ the event that for every $i$, and every $A \in {\cal A}_i$, if ${\cal T}_{i,A}$ is a short subtree then after at most $90 \log k$ balls reach an active subpath of $A$, $A$ has no more active subpaths and in addition, the height of ${\cal T}_{i,A}$ is at most $45 \log k$.
\begin{lemma} \label{l:shortReach}
$\Pr[\mathbf{E_1}] \ge 1-1/k^5$.
\end{lemma}
\begin{proof}
Fix some $i$, and $A \in {\cal A}_i$. Assume that ${\cal T}_{i,A}$ is a short subtree. By definition, $\left\langle i,A \right\rangle$ has no short ancestor.
For $i' = \log_b(\ell/p) \le \log_b(\calD/p)$, $P^*$ itself is short, since $b^{i'} = \ell/p$, and thus all tree nodes in level $i'$ (and lower) are short.  Therefore, $i \le i'$. Since there are at most $k$ nodes in every level of the tree, the number of short subtrees of ${\cal T}$ is at most $k \cdot \log_b(\calD/p) = k \cdot (\log_b \calD + \log_b 100) \le k^4 \log k + O(k \log k) \le O(k^4\log k)$. By the previous lemma, and a union bound over all short subtrees, the result follows.
\end{proof}

Since every node in level $\log _b (\ell/p)$ of the tree belongs to some short subtree, we get the following corollary.
\begin{corollary} \label{c:height}
With probability at least $1-2/k^5$, the height of ${\cal T}$ is at most $\log_b (\ell/p) + O(\log k) \le \iter$.
\end{corollary}

We denote by $\mathbf{E_2}$ the event that for all $i \le \iter$ and $j \in [k]$, the radius of the $j$-th ball of the $i$-th iteration of the outer loop is at most $O(b^i \log k)$.
\ifprocs
From assumption~\ref{a:polyLogDelta} we deduce the following.
\else
\fi
We wish to prove that $\mathbf{E_2}$ holds with high probability. We will need the following lemma, which gives a concentration bound on the sum of independent exponential random variables.
\begin{lemma}\label{l:conc}
Let $X_1,\ldots,X_n$ be independent random variables
such that each $X_j \sim \exp(\lambda_j)$ for $\lambda_j>0$,
and denote $\lambda = \max_j \lambda_j$.
Then $X = \sum_j {X_j}$ 
has expectation $\mu = \EX[X] = \sum_{j}{\lambda_j}$ 
and satisfies
$$
  \forall \delta > 1, \quad \quad \Pr[X > (1+ \delta) \mu] \le e^{(1 - \delta) \frac{\mu}{2 \lambda}} \; .
$$
\end{lemma}
\begin{proof}
We proceed by applying Markov's inequality to the moment generating function
(similarly to proving Chernoff bounds).
Let $j \in [n]$ and consider $0\le t \le \frac{1}{2 \lambda_j}$.
Then the moment generating function of $X_j$ is known and can be written as
$\EX\left[ e^{t X_j} \right] = \frac{1}{1 - t \lambda_j} 
\le 1 + 2 t \lambda_j \le e^{2 t \lambda_j}$.
Now set $t = \frac{1}{2 \lambda}$ and use Markov's inequality to get
\begin{equation*}
\Pr[X > (1+ \delta) \mu] = \Pr[e^{t X} > e^{t(1+ \delta) \mu}] \le \frac{\mathbb{E}\left[e^{t X} \right]}{e^{t(1+ \delta) \mu}} = \frac{\prod_{j \in [n]}{\mathbb{E}\left[ e^{t X_j} \right]}}{e^{t(1+ \delta) \mu}} \le \frac{e^{2t \mu}}{e^{t(1+ \delta) \mu}} = e^{(1 - \delta) \frac{\mu}{2 \lambda}} \;.
\end{equation*}
\end{proof}

\begin{lemma} \label{l:radii}
$\Pr[\mathbf{E_2}] \ge 1 - 1/k^5$.
\end{lemma}
\begin{proof}
Fix $i \le \iter$ and $j \in [k]$. 
In the notations of Algorithm~\ref{alg:part}, let $r_j$ be the radius of the $j$th ball in the $i$th iteration of the outer loop. Then $r_j = \sum_{i' \le i}{R_j^{i'}}$ is the sum of independent exponential random variables.
$\mathbb{E}[r_j] = \sum_{i' \le i}{b^{i'}} = b \cdot \frac{b^i-1}{b-1} \ge 20b^i \log k$. Applying Lemma~\ref{l:conc} we get that 
$$\Pr[r_j > 40b^i \log k ] = \Pr[r_j > 2 \mathbb{E}[r_j] \;] \le e^{\frac{-\mathbb{E}[r_j]}{b^i}} \le k^{-10}$$
By assumption~\ref{a:polyLogDelta}, $\iter = O(\log \calD \log k) = O(k^3 \log k)$.
Thus by a union bound over all values of $i$ and $j$ in question, 
$$\Pr[\forall i ,j. \;  r_j \le 40 b^i \log k \; \; \text{in the $i$th iteration of the outer loop}] \ge 1 - \frac{k \cdot \iter}{k^{10}} \ge 1 - \frac{1}{k^5} \;.$$
\end{proof}

Summing everything up, we can now bound with high probability the weights of all short subtrees of ${\cal T}$. 
\begin{claim} \label{c:shortWeight}
Conditioned on the events $\mathbf{E_1}$ and $\mathbf{E_2}$, for every $i_0 > 1$ and $A \in {\cal A}_{i_0}$, if ${\cal T}_{i_0,A}$ is a short subtree of ${\cal T}$, then the total weight charged to nodes of ${\cal T}_{i_0,A}$ is at most $O(b^{i_0} \log^2 k)$ with certainty.
\end{claim}
\begin{proof}
Conditioned on $\mathbf{E_1}$, at most $90 \log k$ detours are charged to nodes of every short subtree and for $i = i_0 + 45\log k$, there are no more active subpaths of $A$. Conditioned on $\mathbf{E_2}$, the most expensive detour is of weight at most $O(b^{i_0+45\log k}\log k)$, we get that the total weight charged to nodes of the subtree is $90 \log k \cdot O(b^{i_0} \cdot b^{45\log k} \cdot \log k) \le O(b^{i_0} \log^2 k)$, since $b^{45 \log k} = O(1)$. 
\end{proof}
\subsection{Bounding The Weight Of ${\cal T}$}
We are now ready to bound the total weight charged to the tree. 
Recall that for every $i \le \log_b(\ell/p)$ we denoted by $l_i$ the total weight charged to nodes of the form $\left\langle i,A \right\rangle$, where $A \in {\cal A}_i$ is a long active subpath, and by $s_i$ the total weight charged to short subtrees rooted in the $i$-th level of ${\cal T}$. 
Since $\ell \ge 1$, $P^* \in {\cal A}_1$ is long. Therefore, $s_1=0$. We can therefore rearrange Eqn.~\eqref{eq:treeWeight} to get the following.
\begin{equation}
\sum_{\langle i,A \rangle \in {\cal T}}{w_{i,A}} = \sum_{i=1}^{\log_b (\ell/p)}{(l_i+s_{i+1})} \;.
\label{eq:treeWeight2}
\end{equation}

Let $i \le \log_b(\ell/p)$.
Let $A \in {\cal A}_i$ be a long active subpath. That is, $w(A) \ge pb^i$.
Thinking of $A$ as a continuous path, divide $A$ into $w(A) / (pb^i)$ segments of length $pb^i$. Some segments may contain no nodes.
Let $I$ be a segment of $A$, and assume $I$ contains nodes (otherwise, no cost is charged to $A$ on account of detours from $I$). Following Lemma~\ref{l:puncBalls}, we get the following.
\begin{lemma}
With probability at least $1-1/k^{10}$, no more than $90 \log k$ balls reach $I$.
\end{lemma}
Denote by $\mathbf{E_3}$ the event that for every $i \ge \log_b(\ell/k^2)$, and for every long active subpath $A \in {\cal A}_i$, in the division of $A$ to segments of length $pb^i$, every such subsegment is reached by at most $90 \log k$ balls.
\begin{lemma}
$\Pr[ \mathbf{E_3}] \ge 1-1/k^5$.
\end{lemma}
\begin{proof}
Since for every $i \ge \log_b(\ell/p)$, every $A \in {\cal A}_i$ is short, the number of relevant iterations (of the outer loop) is at most $\log_b(\ell/p) - \log_b(\ell/k^2) = \log_b(k^2/p) \le O(\log^2k)$.
For every $i \ge \log_b(\ell/k^2)$ and every long path $A \in {\cal A}_i$, the number of segments of $A$ is at most 
$w(A) / (pb^i) \le w(A) / (p \ell / k^2) \le k^2/p$.
Therefore the number of relevant segments for all $i \ge \log_b(\ell/k^2)$ and for all long $A \in {\cal A}_i$ is at most $O(k^2\log^2k)$
Applying a union bound over all relevant segments the result follows.
\end{proof}
Since $\Pr[\mathbf{E_1}] \ge 1-1/k^5$ and $\Pr[\mathbf{E_2}] \ge 1-1/k^5$, we get the following corollary.
\begin{corollary}
$\Pr[\mathbf{E_1} \wedge \mathbf{E_2} \wedge \mathbf{E_3}] \ge 1-O(k^{-5})$
\end{corollary}
It follows that it is enough for us to prove that conditioned on $\mathbf{E_1}$, $\mathbf{E_2}$ and $\mathbf{E_3}$, with probability $1$ the total weight charged to the tree is at most $O(\log^5k)\ell$.
\begin{lemma} \label{l:long}
Conditioned on $\mathbf{E_2}$ and $\mathbf{E_3}$, $l_i \le O( b^{i} k \log k )$. In addition, if $i \ge \log_b(\ell/k^2)$, 
then $l_i \le O(\log ^2 k) \cdot \ell$ with probability $1$.
\end{lemma}
\begin{proof}
To see the first bound, observe that by the update process of $P$, at most $k$ detours are added to $P$ during the $i$-th iteration. Conditioned on $\mathbf{E_2}$, each one of them is of weight at most $O(b^i \log k)$.
To see the second bound, let $A \in {\cal A}_i$ be a long active subpath.
The additional weight resulting from detours from vertices of $A$ is at most the number of segments of $A$ of length $pb^i$, times the additional weight to each segment. Therefore, the additional weight is at most 
\begin{align*}
\ifprocs
w_{i,A} \le w(A) / p b^i \cdot O( \log k )\cdot O( b^{i} \log k) =  O( \log^2 k ) \cdot w(A)\;.
\else
w_{i,A} \le w(A) / p b^i \cdot O( \log k )\cdot O( b^{i} \log k) =  O( \log^2 k ) \cdot w(A)\;.
\fi
\end{align*}
Since all paths in ${\cal A}_i$ are internally disjoint subpaths of $P^*$, we get:
\begin{align*}
\ifprocs
l_i = \sum_{\longSub \; A \in {\cal A}_i}{w_{i,A}} \le \sum_{\longSub \; A \in {\cal A}_i}{O(\log^2k) \cdot w(A)} \le O(\log^2k) \cdot \ell \;. \quad \quad
\else
l_i = \sum_{\longSub \; A \in {\cal A}_i}{w_{i,A}} \le \sum_{\longSub \; A \in {\cal A}_i}{O(\log^2k) \cdot w(A)} \le O(\log^2k) \cdot \ell \;. \quad \quad
\fi
\qedhere
\end{align*}
\end{proof}

\begin{lemma}\label{l:short}
Conditioned on events $\mathbf{E_1}$, $\mathbf{E_2}$ and $\mathbf{E_3}$, $s_{i+1} \le O(b^{i+1} k\log^2 k)$.  In addition, if $i \ge \log_b(\ell/k^2)$, then $s_{i+1} \le O(\log ^3k)\cdot \ell$.
\end{lemma}
\begin{proof}
Conditioned on $\mathbf{E_1}$ and $\mathbf{E_2}$, we proved in Claim~\ref{c:shortWeight} that the total weight charged to a short subtree rooted in level $i+1$ is at most $O(b^{i+1} \log^2 k)$ with certainty. Since there are at most $k$ such subtrees, the first bound follows.
To get the second bound, note that by the definition of a short subtree, for every short subtree ${\cal T}'$ rooted at level $i+1$, the parent of the root of ${\cal T}'$ consists of a long active subpath $A$ of level $i$. Conditioned on $\mathbf{E_3}$, every segment of $A$ is intersected by at most $90 \log k$ balls. Therefore, $\langle i,A \rangle$ can have at most $(w(A) / p b^i) \cdot 90 \log k$ children, and in particular, children consisting of short active subpaths. The cost of a short subtree rooted in the $i+1$ level of ${\cal T}$ is at most $O(b^{i}\log^2 k)$. Thus the total cost of all short subtrees rooted in children of $A$ is bounded by 
$$(w(A) / p b^i) \cdot 90 \log k \cdot O(b^{i} \log^2 k) \le O(\log^3 k)  \cdot w(A) \;.$$
Summing over all (internally disjoint) long subpaths of level $i$, the result follows.
\end{proof}
We now turn to prove Lemma~\ref{l:totalWeight}.
\begin{proof}[\proofof{Lemma~\ref{l:totalWeight}}]
Since $\Pr[\mathbf{E_1} \wedge \mathbf{E_2} \wedge \mathbf{E_3}] \ge 1-O(k^{-5})$, it is enough to show that conditioned on {\bf E1}, {\bf E2} and {\bf E3}, the total weight charged to the tree is at most $O(\log ^5 k) \ell$ with certainty.
Recall that $\sum_{\langle i,A \rangle \in {\cal T}}{w_{i,A}} = \sum_{i=1}^{\log_b (\ell/p)}{(l_i+s_{i+1})}$. Following Lemmas~\ref{l:long} and \ref{l:short} we get that
\ifprocs
\begin{equation*}
\sum_{i=1}^{\log_b (\ell/k^2)}{(l_i+s_{i+1})} \le o(1) \cdot \ell \;.
\end{equation*}
\else
\begin{equation*}
\begin{split}
\sum_{i=1}^{\log_b (\ell/k^2)}{(l_i+s_{i+1})} &\le \sum_{i=1}^{\log_b (\ell/k^2)}{O( b^{i} k \log k + b^{i+1} k\log^2 k)} \\
&\le O(k \log^2 k)\sum_{i=1}^{\log_b (\ell/k^2)}{b^i} = O(k \log^2 k) \frac{b}{b-1} \cdot \frac{\ell}{k^2} = o(1) \cdot \ell \;.
\end{split}
\end{equation*}
\fi
In addition,
\ifprocs
\begin{equation*}
\sum_{i=\log_b (\ell/k^2) + 1}^{\log_b (\ell/p)}{(l_i+s_{i+1})} \le O(\log^5 k) \cdot \ell \;.
\end{equation*}
\else
\begin{align*}
\sum_{i=\log_b (\ell/k^2) + 1}^{\log_b (\ell/p)}{(l_i+s_{i+1})} &\le \sum_{i=\log_b (\ell/k^2) + 1}^{\log_b (\ell/p)}{O(\log^2k) \ell + O(\log^3k) \ell} \\
&\le O(\log^3 k) \ell \cdot (\log_b(\ell/p) - \log_b(\ell/k^2)) \\
&\le O(\log^3 k) \ell \cdot O(\log^2k)  = O(\log^5 k) \cdot \ell \;.
\qedhere
\end{align*}
\fi
\end{proof}

\section{Terminal-Centered Minors: Extension to General Case}\label{sec:proofDiscard}

In this section we complete the proof of Theorem~\ref{thm:main} by reducing it to the special case where Assumption~\ref{a:polyLogDelta} holds (which we proved in Section~\ref{sec:TCM}).
We first outline the reduction, which is implemented using a recursive algorithm,
as follows. 
The algorithm initially rescales edge weights of the graph so that minimal terminal distance is $1$. If $\calD < 2^{k^3}$ then we apply Algorithm~\ref{alg:part} and we are done. Otherwise, we construct a set of at most $k-1$ low-diameter balls which are mutually far apart, and whose union contains all terminals. Then, for each of the balls, we apply Algorithm~\ref{alg:part} on the graph induced by that ball. 
Each ball is then contracted into a ``super-terminal". We apply the algorithm recursively on the resulting graph $\tilde{G}$ with the set of super-terminals as the terminal set. Going back from the recursion, we ``stitch" together the output of Algorithm~\ref{alg:part} on the balls in the original graph with the output of the recursive call on $\tilde{G}$, to construct a partition of $V$ as required.
The detailed algorithm and proof of correctness  
\ifprocs
appear in the full version \cite{KKN13}.
\else
are described in Section~\ref{sec:general-alg}. Before that, we need a few definitions.

Assume that the edge weights are already so that the minimum inter-terminal distance is 1. Denote by $D$ the set of all distances between terminals, rounded down to the nearest powers of $2$. Note that $|D| < k^2$. Consider the case $\calD > 2^{k^3}$. There must exist $0 \le m_0 \le k^3-k$ such that $D \cap \{2^{m_0},2^{m_0+1},\ldots,2^{m_0+k}\} = \emptyset$.
Define $R \eqdef \{(x,y) \in T^2 : \ d_G(x,y) < 2^{m_0} \}$. 
\begin{claim} \label{c:equivR}
$R$ is an equivalence relation.
\end{claim}
\begin{proof}
Reflexivity and symmetry of $R$ follow directly from the definition of a metric. To see that $R$ is transitive, let $x,y,z \in T$, and assume $(x,y),(y,z) \in R$. Therefore $d_G(x,y) < 2^{m_0}$ and $d_G(y,z)<2^{m_0}$. By the triangle inequality, $d_G(x,z)<2^{m_0+1}$. Since $D \cap \{2^{m_0},\ldots,2^{m_0+k}\} = \emptyset$, $d_G(x,z)<2^{m_0}$, and therefore $(x,z)\in R$.
\end{proof}
For every equivalence class $U \in T\raisebox{-.1em}{/}\raisebox{-.4em}{R}$, we pick an arbitrary $u \in U$, and define $\hat{U} = B_G(u,2^{m_0})$.
\begin{claim}\label{c:sepBalls}
${\cal U} \eqdef \{\hat{U}\}_{U \in T/_R}$ is a partial partition of $V$. Moreover, for every $U \in T\raisebox{-.1em}{/}\raisebox{-.4em}{R}$, $U \subseteq \hat{U}$, $G[\hat{U}]$ is connected and of diameter at most $2^{m_0+1}<2^{k^3}$.
\end{claim}
\begin{proof}
Let $U \in T\raisebox{-.1em}{/}\raisebox{-.4em}{R}$. Let $u \in U$ be such that $\hat{U}=B_G(u,2^{m_0})$. For every $x \in U$, by the definition of $R$, $d(x,u) < 2^{m_0}$, and thus $x \in \hat{U}$. Therefore $U \subseteq \hat{U}$. By the definition of a ball, $G[\hat{U}]$ is connected and of diameter at most $2^{m_0+1}<2^{k^3}$.
To see that ${\cal U}$ is a partial partition of $V$, take $U' \in T\raisebox{-.1em}{/}\raisebox{-.4em}{R}$ such that $U \ne U'$, and let $u' \in U'$ be such that $\hat{U'}=B_G(u',2^{m_0})$. Since $(u,u') \notin R$, $d_G(u,u') \ge 2^{m_0}$, and since 
$D \cap \{2^{m_0},\ldots,2^{m_0+k}\} = \emptyset$, $d_G(u,u') \ge 2^{m_0+k+1}$, thus $\hat{U} \cap \hat{U'} = \emptyset$.
\end{proof}

\subsection{Detailed Algorithm}\label{sec:general-alg}

\begin{algorithm}[t]
\caption{Partitioning $V$ - The General Case}
\label{alg:gen}
\begin{algorithmic}[1]
\REQUIRE $G = (V,E,w),\; T =\{t_1,\ldots,t_k\}\subseteq V$
\ENSURE A partition $\{V_1,V_2,\ldots,V_k\}$ of $V$.
\STATE rescale the edge weights so that the minimal terminal distance is $1$.
\IF {$\calD \eqdef \max_{u,v \in T}d_G(u,v) \le 2^{k^3}$}
\STATE run Algorithm~\ref{alg:part}, and return its output.
\ELSE
\STATE define $R$ and ${\cal U}$ as above.
\FORALL{$\hat{U} \in {\cal U}$}
\STATE run Algorithm~\ref{alg:part} independently on $G[\hat{U}]$.
\STATE contract $\hat{U}$ to a single ``super-terminal'', maintaining edge weights of all remaining edges. \label{al:super}
\ENDFOR
\STATE denote the resulting graph $\tilde{G}$. 
\STATE run Algorithm~\ref{alg:gen} recursively on $\tilde{G}$ with the set of super-terminals.
\FORALL {super-terminals $u \in \tilde{G}$}
\STATE let $u_1,\ldots,u_r$ be the terminals contracted to $u$ in line~\ref{al:super} in an arbitrary order.
\FORALL {vertices $v$ assigned to $u$ in the recursive call}
\STATE assign $v$ to its nearest terminal among $u_1,\ldots,u_r$. \\
Break ties by the ordering of $u_1,\ldots,u_r$. \COMMENT{Making sure we construct a minor.}
\ENDFOR
\ENDFOR
\ENDIF
\RETURN the resulting partition of $V$.
\end{algorithmic}
\end{algorithm}

Our algorithm for the general case of Theorem~\ref{thm:main} is given 
as Algorithm~\ref{alg:gen} (which makes calls to Algorithm~\ref{alg:part}).
It is clear that this algorithm returns a partition of $V$. 
In addition, since every level of recursion decreases the number of terminals in the graph, the depth of the recursion is at most $k$. During each level of the recursion, Algorithm~\ref{alg:part} is invoked at most $k$ times. Therefore, Algorithm~\ref{alg:part} is invoked at most $k^2$ times, each time on a set of at most $k$ terminals. Note that the result of Lemma~\ref{l:totalWeight} still applies if $k$ is only an upper bound on the number of terminals, and not the exact number of terminals. Therefore, we get that there exists $C_0>0$, such that all $O(k^2)$ times that the algorithm is invoked, it achieves a weight stretch factor of at most $C_0\log^5k$, with probability at least $1 - O(k^{-3})$. Applying a union bound, we get that with high probability, the stretch bound is obtained in all invocations of the algorithm. It remains to show that this suffices to achieve the desired stretch factor in $G$. 
\begin{lemma} \label{l:discard}
With probability at least $1 - 1/k$, on a graph with $\tau \le k$ terminals, 
Algorithm \ref{alg:gen} obtains a stretch factor of at most $C_0 \log^5 k + \log^5 k \cdot 2^{-k} \sum_{k' \le \tau}{2(k')^2} \le 2C_0 \log^5k$.
\end{lemma}

\begin{proof}
It is enough to show that conditioned on the event that every invocation of Algorithm~\ref{alg:part} achieves a stretch factor of at most $C_0 \log^5k$, the generalized algorithm achieves the desired stretch factor.
We prove this by induction on $k$. For the case $k=2$, rescaling the weights assures that $\calD=1 \le 2^{k^3}$, and therefore Algorithm~\ref{alg:part} is applied on $G$. The result follows from the proof of Section~\ref{sec:TCM}. Assuming correctness for every $\tilde{k} < k$, we prove correctness for $k$. 
Let $s,t \in T$. If $(s,t) \in R$, then $s$ and $t$ are in the same set in ${\cal U}$, and by the conditioning, the stretch factor of the distance between $s$ and $t$ is at most $C_0 \log^5k$. This does not change in steps $4-7$ of the algorithm.
Otherwise, Let $\tilde{s},\tilde{t}$ be the terminals (or super-terminals) associated with $s$ and $t$ in $\tilde{G}$ respectively.
Denote $d=d_{G,w}(s,t)$ and $\tilde{d} = d_{\tilde{G},\tilde{w}}(\tilde{s},\tilde{t})$. Denote by $G'=(V',E',w')$ the terminal-centered minor induced by the partition returned by the recursive call, and by $G''=(V'',E'',w'')$ the terminal-centered minor induced by the partition returned in the final step of the algorithm. 
Denote $d'=d_{G',w'}(\tilde{s},\tilde{t})$. and $d''=d_{G'',w''}(s,t)$.
Let $u$ be a super terminal on a shortest path $P'$ between $\tilde{s}$ and $\tilde{t}$ in $G'$.
Let $P''$ be the path obtained from $P'$ in $G''$ in the following manner. In the place of every super-terminal $u$ in $P'$, originating in some node set $\hat{U}$, we add a path between the corresponding terminals in $\hat{U}$ (based on the terminal-centered minor constructed for $G[\hat{U}]$ in step $3$). The edges of $P'$ are also replaced with corresponding edges in $G''$.

Recall that in $G'$, the weight of every edge is the distance between its endpoints in $\tilde{G}$ (by the definition of a terminal-centered minor). In $G''$ the weight of every edge is the distance between its endpoints in $G$. Therefore the weight 
$P'$ contained at most $k-1$ edges. In $G''$ the weight of each such edge increases by at most $k 2^{m_0}$. In addition, every expansion of a super-terminal adds at most $k 2^{m_0}$ to the path. Therefore, $w''(P'') \le w'(P') + 2k^22^{m_0} \le d' + d \cdot 2k^22^{-k}$.
By the induction hypothesis $$d' \le \left(C_0 \log^5 k + \log^5 k \cdot 2^{-k} \sum_{k' \le k-1}{(k')^2}\right)\tilde{d}.$$ 
Since $\tilde{d} \le d$, we get that 
\ifprocs
$d''$ is at most
\begin{equation*}
\left(C_0 \log^5 k + \log^5 k \cdot 2^{-k} \sum_{k' \le k}{2(k')^2}\right)d.
\end{equation*}
\else
\begin{equation*}
\begin{split}
d'' &\le \left(C_0 \log^5 k + \log^5 k \cdot 2^{-k} \sum_{k' \le k-1}{2(k')^2}\right)d + d \cdot 2k^22^{-k} \\
&\le \left(C_0 \log^5 k + \log^5 k \cdot 2^{-k} \sum_{k' \le k}{2(k')^2}\right)d.
\end{split}
\end{equation*}
\fi
This completes the proof of Lemma~\ref{l:discard}
(and in fact also of Theorem~\ref{thm:main}).
\end{proof}
\fi

\paragraph{Acknowledgments.}
We thank the anonymous reviewers for useful comments, 
and particularly for suggesting to improve the proof of Lemma~\ref{l:radii}
using a tail bound for sum of exponential random variables, 
which saves a factor of $\log k$ in our main result, Theorem~\ref{thm:main}.

\bibliographystyle{alphaurlinit}
\bibliography{robi}
\end{document}